\documentclass[12pt,onesided]{article}

\usepackage{amsmath}
\usepackage{amsthm}
\usepackage{amssymb}
\usepackage{amsfonts}
\usepackage{array}
\usepackage{curves}
\usepackage{epic}
\usepackage{color}
\usepackage{hyperref}
\usepackage{verbatim} 
\usepackage{rotating}

\newcommand{\R}{\mathbb{R}}

\newcommand{\E}{\mathbf{E}}

\newcommand{\F}{\mathcal{F}}

\newtheorem{theorem}{THEOREM}

\newtheorem{lemma}{LEMMA}
\newtheorem{definition}{DEFINITION}
\newtheorem{assumption}{ASSUMPTION}

\newcounter{remark}

\def\bi{\begin{itemize}}
\def\ei{\end{itemize}}
\def\i{\item}

\begin{document}

\title{No-arbitrage pricing under cross-ownership}
\author{Tom Fischer\thanks{Institute of Mathematics, University of Wuerzburg, Am Hubland, 97074 Wuerzburg, Germany.
Tel: +49 931 31 88911.
E-mail: {\tt tom.fischer@uni-wuerzburg.de}.
}\\
{University of Wuerzburg}} 
\date{First version: January 9, 2010\\
This version: \today}
\maketitle

\begin{abstract}
We generalize Merton's asset valuation approach to systems of multiple financial firms 
where cross-ownership of equities and liabilities is present.
The liabilities, which may include debts and derivatives, can be of differing seniority. 
We derive equations for the prices of equities and recovery claims under no-arbitrage.
An existence result and a uniqueness result are proven.
Examples and an algorithm for the simultaneous calculation of all no-arbitrage prices are provided. 
A result on capital structure irrelevance for groups of firms regarding externally held claims 
is discussed, as well as financial leverage and systemic risk caused by cross-ownership. 

\end{abstract}




\tableofcontents



\section{Introduction}

\subsection{Preliminaries}
\label{Preliminaries}

The main purpose of this paper is to investigate and overcome the intricacies of 
determining no-arbitrage prices for the equities and
the liabilities of a group of firms where the financial fates of these 
firms are intertwined through the cross-ownership of financial assets issued or guaranteed by the
same firms. In this theory, the liabilities can include debt and derivatives,
and the claims belonging to these liabilities are allowed to
have differing priorities in a potential liquidation.
 
On the one hand, the problem of no-arbitrage pricing under cross-ownership
is real, and very important to consider, since cross-ownership of 
financial assets is present in the world's financial markets. For instance, McDonald (1989) writes
that in 1987 double counting from {\em cross-ownership of equity} (also called ``reciprocal ownership", 
``corporate cross-holding", or ``intercorporate shareholding")
accounted for ``at least 24\% of Japan's reported market capitalization." 
An article by B{\o}hren and Michalsen (1994) shows at the example of the Oslo stock exchange how 
equity cross-ownership can lead to double counting of assets and overstated market equity. 
Ritzberger and Shorish (2002) state 
``that under cross-ownership the book value of a firm will tend to be overestimated with respect to the
underlying cash flows." B{\o}hren and Michalsen (1994) also show that 
financial leverage can be understated due to cross-ownership.
Furthermore, given the evidence from more recent events in global financial markets,
not only the danger of overstated market aggregates, but also the very real risk of 
financial contagion seems to at least partially stem from cross-ownership. 
However, especially in this context of systemic risk, cross-ownership should not only be 
considered for equities, but also for liabilities like bonds and derivatives issued by the
considered firms.

On the other hand, the topic of financial cross-ownership is essentially
non-existent in the literature of financial mathematics as far as asset valuation models, 
credit risk models, and derivatives pricing are concerned. In financial economics, or finance in general,
cross-ownership {\em is} considered, but mostly in conjunction with the separation of ownership from
the control of firms (e.g.~in  Bebchuk, Kraakman and Triantis (2000), Ritzberger and Shorish (2002), 
Dorofeenko et al.~(2008), and references therein), or with respect to the distortion of market aggregates 
(e.g.~McDonald (1989), B{\o}hren and Michalsen (1994), and references therein). 
Because of this, the focus of cross-ownership considerations in financial economics primarily 
lies on equity. This, however, is insufficient from a mathematical finance point of view, because
it neglects aspects of liabilities like debt or derivatives, where cross-ownership
can play a substantial role in the context of asset valuation under counterparty risk.

It is indeed somewhat surprising that mathematical finance has neglected the area of cross-ownership,
since, for instance when trying to assess the credit-worthiness of a financial firm, one 
of the natural questions to ask would be: what does it mean for the balance sheet
of company A if company B defaults on its debt, given that A owns parts of B's debt and maybe 
even some of its equity? Also, how severely are the financial promises of A (debt or derivatives
issued by A) affected when B defaults? For a larger number of firms that
are intertwined by cross-ownership, these kind of questions can quickly become very complicated.  
As Ritzberger and Shorish (2002) write in the context of separation of ownership from
the control of firms:
\begin{quote}
[...] The precise quantitative effect of cross-ownership between firms is, however, difficult
to capture, both at the theoretical and the empirical level. [...]

[...](``A owns part of B, B owns part of A, so A owns part of B's ownership of A, which is also 
part of a part of A's ownership of B, which is..."). This recursion must be addressed, [...]
\end{quote}
In our context of asset valuation, the ``recursion" Ritzberger and Shorish write about could 
be thought of as a self-feeding financial feedback loop -- potentially a financial vicious circle. 
For instance, in the earlier example, the deterioration of A, that was originally caused 
by B, could damage the financial status of B even further. It is clear that these kind of
scenarios have to be considered when it comes to pricing financial assets. The questions 
asked above go straight to the core of a very important and very timely problem: 
{\em financial contagion} and, indeed, {\em systemic risk}.

\subsection{Asset valuation models}

To understand why cross-ownership has virtually never been considered in mathematical finance,
one needs to have a closer look at modern credit risk models. The reason for this is that 
Merton's asset valuation model, also called the {\em Merton model} 
(Black and Scholes (1973), Merton (1973, 1974)) and its offspring today are mostly 
used in the context of credit risk management. 

Merton's model, where equity and debt are considered as 
derivatives of an underlying value process of the firm's assets, 
inspired many extensions and refinements.
These models are usually summarily called {\em structural} models, also ``firm-value models"
or ``threshold models", since they attempt to value assets by modelling the financial structure of the
considered companies. In these credit risk models, a credit event is usually triggered when the assets
(the ``firm value") of a company fall below a certain threshold -- in Merton's case the nominal amount of
outstanding debt. This fundamental idea is the basis of a plethora of models that aim to improve the 
original Black-Merton-Scholes approach. Already Merton (1974) includes coupon paying bonds and explains
how to incorporate stochastic interest rates. Further models that extend the original idea to 
stochastic interest rates,   
differing liabilities, 
differing maturities,   
counterparty risk (``vulnerable options"),  
and jump diffusion   
exist - including multi-firm models.   
It is far beyond the scope of this paper to give a comprehensive overview of these models.
For a summary of the literature we refer to standard textbooks like Bingham and Kiesel (2004),
as well as the references therein.

In contrast to the structural approaches, so-called {\em reduced-form} models
attempt to model credit risk and especially default rates (therefore also
the name ``default rate models") in a  more statistical way, usually by assuming independence
of credit events conditional on certain underlying stochastic factors. 
For an overview over credit risk models and over structural versus reduced-form models see for
instance Crouhy, Galai and Mark (2000) and Arora, Bohn and Zhu (2005), or textbooks like
Bingham and Kiesel (2004). In general, it is clear that 
default rate models are less useful if asset valuation (rather than credit risk
management) for a broad spectrum of assets and liabilities is the main goal. 

While the Black-Merton-Scholes approach is the basis of modern
structural credit risk models (the model of Moody's KMV 
(cf.~Arora, Bohn and Zhu (2005) and Kealhofer and Bohn (2001)) 
possibly being the commercially most successful one), 
a rupture in the rationale of structural approaches seems to appear as soon as multiple 
firms are considered. While at the level of the individual firm the blueprint of the Merton model
is usually clearly visible, at the multi-firm or inter-firm level, structural appproaches 
seem to not get any consideration at all. For instance, models like the one of Moody's KMV seem
to not go far beyond modelling correlations (cf.~Kealhofer and Bohn (2001)) or, more generally, 
multivariate distributions for the firms' underlying assets.
Such dependencies between assets certainly lead to 
dependencies between the firms' default events when the original Black-Merton-Scholes idea is then 
applied to the individual firm -- which is that default is triggered when the assets of the company
fall below a certain level. However, it is obvious that such
approaches do not attempt to model the actual ownership or cross-ownership structures that possibly 
contribute to the observed credit event correlations. In this sense, it is possibly fair 
to say that structural models turn reduced-form at the multi-firm or inter-firm level.

Before moving on to the next section, in the context of contagion at the two-firm level, 
i.e.~regarding counterparty risk or ``vulnerable options", 
there are certainly many papers considering this problem. For instance, see Jarrow and Turnbull (1995)
and the references therein. Regarding financial contagion at a larger scale, it should be mentioned
that there are articles on credit risk that focus especially on modelling credit contagion 
(e.g.~Horst (2007)). However, the question of (and the role of) multi-firm, multi-liability, and
multi-priority cross-ownership is generally not addressed in these papers.

\subsection{General liabilities under cross-ownership}
\label{General liabilities under cross-ownership}

While the focus of the Merton model and its (structural) sucessors lies on credit risk, this is not 
the main focus of this paper. Let us assume that we had a `correct' model of all considered firms' 
financial structures which was also correctly describing the inter-firm structures in the form of 
cross-ownership -- cross-ownership not only of equity, but also of general liabilities that besides
debt could be in the form of derivatives that had been issued by or that were guaranteed by these firms.
One would expect that such a `correct' model should not only make it possible to model credit events 
more precisely. Additionally, better pricing of equities and liabilities in general should be possible.
Similarly, it should be possible to properly value an issued derivative {\em including}
any counterparty risk present. 

In reality, and in contrast to most existing models, liabilities can be of differing seniority. 
The {\em seniority} of a liability, or the {\em priority} of the corresponding
financial claim, defines the order of repayment in a potential liquidation event.
For instance, senior debt must be paid before subordinate debt, where subordinate debt itself
can again be ranked by so-called tiers. In general, equity in the form of common stock is subordinate
to other liabilities, i.e.~equity is the ``residual claim". It is therefore clear that the 
proper incorporation of the seniority of liabilities is important when it comes to valuation. 

The general liabilities mentioned above should not only include debts or simple financial 
derivatives. In general, a liability could be any properly defined financial 
commitment. For example, it could be the (naked) short-sale of another asset or financial claim, 
a liability or a guarantee related to a mortgage-backed security, as well as an insurance liability. 
For instance, Walsh (2009) writes in The New York Times:
\begin{quote}
[...] A.I.G.'s individual insurance companies have been doing an unusual volume of business with 
each other for many years -- investing in each other's stocks; borrowing from each other's
investment portfolios; and guaranteeing each other's insurance policies, even when they have
lacked the means to make good. [...]
\end{quote}
Here would be a real life situation where cross-ownership was (or is) possibly not only present in equities, 
but also in insurance liabilities or derivatives thereof -- additionally to any other financial
liabilities present, e.g.~in the form of derivatives contracts with firms outside A.I.G.
As far as not obvious anyway, this case clearly shows that cross-ownership is not only an issue
in national Asian or European markets (cf.~McDonald (1989) for Japan, B{\o}hren and Michalsen (1994) for 
Norway; also compare the concept of the {\em Deutschland AG}, a historic network of cross-ownership 
or ``capital entanglement" among Germany's blue chip companies (cf.~H\"opner and Krempel (2003))), 
but it can very much also be
an issue within a conglomerate of firms or within a holding company.
As such, cross-ownership issues should have a very high priority for any regulator or regulatory
framework like Basel II or Solvency II.

The importance of cross-ownership goes beyond the examples mentioned so far.
Cross-ownership is present in more markets and at more levels than first meets the eye.
For instance, consider a houseowner who serves a mortgage on her house. We can consider the houseowner
as a financial firm, where the house is possibly the largest part of the assets, and the mortgage 
possibly is the main liability. 
Like a firm, the houseowner has a limited liability (at least in the U.S.) since she could
declare personal bankruptcy or default on her mortgage should her total assets become worth less than
her total liabilities. 
Assume now that the mortgage holder also owns a portfolio of stocks, for instance through her
pension fund. Furthermore assume that the mortgage was arranged by a major investment bank, was 
sliced and packed into a mortgage-backed security which was sold by the investment bank, possibly together
with a guarantee against default. It is now not at all beyond imagination that the pension fund
of our mortgage holder not only owns shares in the mentioned investment bank, but maybe it also
holds mortgage-backed securities that might include our homeowner's mortgage, plus guarantee. 
On a nationwide or
even global scale, it is therefore entirely imaginable that homeowners who owe a mortgage might in
turn indirectly own parts of their own mortgage and possibly parts of other people's mortgages,
including the corresponding guarantees. 
They also might own shares in the involved investment banks that possibly still hold parts of
these mortgages too.
The implications are clear: a downturn in the housing market (for whatever reason) might not
only affect the homeowner's equity directly through the house ownership, but it might also affect her
by owning the investment banks or mortgage-backed securities involved. This would affect the
value of the mortgage, which in turn influences the value of the investment bank 
(if it holds any) and the value of any
mortgage-backed securities and guarantees that are involved, etc. There is a clear risk 
of contagion or even a systemic risk present in such a situation. A realistic model should try
to measure these risks and take them into account when pricing equities and liabilities.

The examples above should be sufficient to understand why modelling
cross-ownership relationships and properly taking them into account when pricing equities 
and liabilities should be a major priority of mathematical finance. 
A generalized multi-firm structural model that 
incorporates cross-ownership structures should have a profound impact not only on the theory, but
also on the practice of asset valuation. It is clear that it can be difficult to
obtain cross-ownership information and that optimally one would like to have perfect information
(lack of information, especially of price information about the underlying assets, is a 
general problem of all structural models).
It is also clear that in many cases influences other than cross-ownership might play large
roles in (or even dominate) the pricing procedure as well. 
However, some observers, like regulators, might have more information than ordinary market
participants, and it is clear that any properly used cross-ownership information at all can
improve a model by making it more realistic than one where such information is ignored. 
In that sense, any cross-ownership information is valuable.

\subsection{Valued added}
\label{Valued added}

The model of this paper is an extension of the Merton model in the sense that it has only one
maturity date and in the sense that prices at maturity are determined assuming
no-arbitrage, which includes the assumption of no-arbitrage in a potential liquidation event.
However, the main focus of the paper is 
not the calculation of no-arbitrage prices prior to maturity, but the determination
of no-arbitrage prices {\em at} maturity, which is trivial in the Merton model, but not in 
the cross-ownership case.
It will become clear that the calculation of no-arbitrage prices before maturity is straightforward by
risk-neutral pricing (cf.~Sec.~\ref{Equilibria for non-zero maturities}) once no-arbitrage prices 
at maturity for any given scenario of the underlying exogenous assets have been determined. 
The ways in which our model extends the theory of asset valuation in Merton's model are the following: 
\begin{enumerate} 
\item It models the cross-ownership of assets and liabilities for groups of firms.
\item Multiple different classes of liabilities, e.g.~debt and derivatives, are allowed.
\item Multiple different priorities of claims in a possible liquidation are allowed.
\item It directly incorporates counterparty risk at the multi-firm level into derivatives pricing.
\item Underlying exogenous assets can be stochastically dependent, a feature other existing 
multi-firm models have.
\item A principle of capital structure irrelevance for a system of firms in a (no-arbitrage) price 
equilibrium is derived (value of externally held claims = value of exogenous assets).
\item Situations where valuation is impossible (no no-arbitrage price equilibrium/multiple no-arbitrage 
price equilibria) can exist.
\i It should be possible to extend any model based on Merton's original idea 
(for instance the Moody's KMV model) to incorporate 
cross-ownership along the lines of the model of this paper.
\end{enumerate} 
The main challenge in this paper is to formulate
and solve the (in Merton's single firm model: trivial) no-arbitrage equations 
(cf.~Sec.~\ref{The equilibrium equations}) that properly account for cross-ownership in terms of
the balance sheet items of all firms at maturity, under any given possible economic scenario 
of the underlying exogenously priced assets. Why this is trivial in Merton's model
but much less so in the extended model will become clear in the simple example of Section
\ref{Cross-ownership: a two-firm example}. A further important insight will be that
in the case of a unique no-arbitrage price equilibrium, similarly to the Merton model, 
{\em all equities and liabilities are direct derivatives of the underlying exogenously priced assets}. 
This makes no-arbitrage pricing at times before maturity possible by means of risk-neutral valuation
techniques. 

The paper is structured as follows.
After the motivating example of Section \ref{Cross-ownership: a two-firm example} and the
introduction of some notation in Section \ref{Notation and mathematical preliminaries},
Section \ref{Equilibrium for liabilities with multiple priorities} presents and interprets the main
results. These results are essentially contained in Theorem \ref{master-theo}, which is on the
existence and uniqueness of no-arbitrage price equilibria under the cross-ownership of equities 
and general liabilities
of differing seniority. Section \ref{Equilibrium for liabilities with multiple priorities} also
derives results for the accounting equations of the considered system of firms. Measures of the degree
of cross-ownership with regards to financial leverage and financial contagion are defined and discussed.
After the proof of Theorem \ref{master-theo} in Section \ref{Proof and algorithm}, we consider
several applications and examples in Section \ref{Examples}. A conclusion and a technical appendix 
follow.


\section{Cross-ownership: a two-firm example}
\label{Cross-ownership: a two-firm example}

To illustrate and motivate the main results of this paper, we will have a look at an example
with only two companies. Assume that company $i$ ($i=1,2$) has outstanding nominal debt of $b_i \geq 0$.
There are no coupon or interest payments, and both loans have to be paid back at the same future point
in time (maturity). 
Company $i$ holds assets of market value $a_i\geq 0$ at maturity. We assume that these assets are exogenously 
priced in the sense that the capital structure of the companies $i=1,2$ has no
influence on the value of these assets. 
In the original Merton model, the equity of company $i$ at maturity, $s_i$, 
is now given by
\begin{equation}
\label{s_i}
s_i = (a_i - b_i)^+ = \max\{0, a_i - b_i\} .
\end{equation}
The reason for this valuation formula is that, under the bankruptcy laws of many economies,
equity in firm $i$ essentially is a European call option on the assets of $i$, $a_i$, with a strike price
of $b_i$. Expressed in layman's terms, the owners of $i$ get what is left of the
company's assets after the debt $b_i$ has been paid. 
However, because of limited liability, the owners will never encounter negative equity, 
i.e.~the owners will never have to make up for losses of the creditors of firm $i$ by paying
them out of their own pockets after their (the shareholders') stake in the equity of $i$ has become
worthless. We can therefore say that $s_i$ is the {\em liquidation value}
(in a perfectly liquid market with no frictions like transaction costs or taxes)
of the equity (``common stock") of company $i$ at maturity. The {\em no-arbitrage price} 
of company $i$ {\em at maturity} should therefore 
be identical to $s_i$ as otherwise arbitrage opportunities would exist.
The creditors who gave company $i$ the loan with principal $b_i$ at outset recover at maturity
the amount
\begin{equation}
\label{r_i}
r_i = \min\{b_i, a_i\} .
\end{equation}
From the creditors' point of view, $r_i$ is {\em recoverable part of the claim} $b_i$ 
(also called {\em recovery claim}).
From company $i$'s point of view, $r_i$ is the {\em payable part of the liability} $b_i$.
The reason for \eqref{r_i} is that the creditors can not collect more than the market value of the 
company's assets, $a_i$, even if $a_i$ is less than the outstanding principal $b_i$. Therefore, $r_i$ 
can be interpreted as the liquidation value of the nominal debt $b_i$ at maturity, and hence it 
must be identical to the price of this debt at maturity under the assumption of no arbitrage.
Actual market prices and no-arbitrage prices might converge in real markets (even in
the absence of a realistic chance of immediate liquidation) for the reason of so-called capital-structure
arbitrage where hedge funds ``attempt to use these models to buy the underpriced part of a firm's capital
structure, be it debt or equity, and sell the overpriced part'' (Robert C.~Merton in Mitchell (2004)).

As an example, if $b_i$ was 100, and at maturity $a_i$ was 150, then equity would be $s_i = 50$ and the
loan recovery $r_i = 100$, i.e.~the full loan would be recovered. In contrary, was $b_i=100$ as before, 
but at maturity $a_i$ was 50, then equity would be $s_i = 0$ since the owners of the firm would not 
exercise their call option on the de-facto negative equity. The loan recovery would 
therefore only be $r_i = a_i = 50$. 

In the classical Merton model, the (by $a_i$ and $b_i$) uniquely determined no-arbitrage prices at maturity,
\eqref{s_i} and \eqref{r_i}, are used to obtain no-arbitrage prices for stocks and bonds at times 
{\em before} maturity by modelling the price process of the assets ($a_i$ at maturity) before maturity by a 
geometric Brownian motion. The Black-Scholes model (Black and Scholes (1973), Merton (1973)) then
provides a direct solution for the no-arbitrage prices of equities and debts (Merton (1974)). 
The {\em crucial non-stochastic ingredients} of this pricing approach are the no-arbitrage prices at
maturity given by \eqref{s_i} and \eqref{r_i} which show that equity and debt are mere derivatives of the
underlying exogenous assets.

Expressed at the level of our example with two firms, the question this paper will investigate is: what
happens to the no-arbitrage prices of stocks and bonds at maturity, $s_i$ and $r_i$, if company
$i$ ($i=1,2$) is allowed to own stock or bonds of company $j$ ($j=2,1$)? In other words, can (and if so,
how) the Merton model be generalized to the case where cross-ownership is allowed?

This is a very important question for two reasons. First and foremost, as pointed out before,
cross-ownership is present in real markets, therefore it has to be modelled appropriately. 
At the time of writing this paper, we are not aware of any extension
or generalization of Merton's corporate debt model in this direction. 
Second, if under cross-ownership there were unique no-arbitrage prices at maturity as
in the case of the classical Merton model without cross-ownership, then, because of the existence
of an implicit function that mapped exogenous asset prices on endogenous assets' no-arbitrage prices,
the calculation of no-arbitrage
prices at times other than maturity would be straightforward from the completeness of the original
Black-Scholes approach used in Merton (1974). The reason for this is that as in
the case of the classical Merton model, the no-arbitrage prices of equities and debts at maturity would
simply be derivatives of the underlying exogenous assets and they could therefore be priced by the 
risk-neutral pricing approach using an equivalent martingale measure. For this reason, the present paper
will focus on the investigation of {\em no-arbitrage prices at maturity}, since at least from a theoretical
perspective, no-arbitrage prices before maturity are straightforward to obtain if unique price equilibria
at maturity exist and if a complete model for the price processes of the exogenous assets is chosen.

For an illustration of the changed situation under cross-ownership again consider the above two-firm
example ($b_i=100$ for $i=1,2$), but with the following changes to its set-up: assume that company $i$ ($i=1,2$)
owns 50\% of the equity of company $j$ ($j=2,1$). We assume here that ownership in any of the two
companies is an homogeneous asset class where all owners have the same rights in proportion.
Assume now that there were no-arbitrage prices for equities and debts, $s_i$ and $r_i$ ($i=1,2$). 
As before in the case with no cross-ownership, we define no-arbitrage prices as no-arbitrage prices
under liquidation. 
If this is the case, equation \eqref{s_i} turns into ($i=1,2$; $j=2,1$)
\begin{equation}
\label{s_i_c}
s_i =  (a_i + 0.5s_j - b_i)^+ ,
\end{equation}
since the total assets of company $i$ ($i=1,2$) are given by the exogenous assets $a_i$ plus the share in the
equity of company $j$ ($j=2,1$), $0.5s_j$. The recoverable debt at maturity, $r_i$, is
\begin{equation}
\label{r_i_c}
r_i = \min\{b_i, a_i + 0.5s_j\} .
\end{equation}
The fundamental difference to the set-up with no cross-ownership in equations \eqref{s_i} and \eqref{r_i}
is that \eqref{s_i_c} is not a direct expression of $s_i$ in terms of the exogenous assets $a_i$ and the debt
level $b_i$. Instead, \eqref{s_i_c} is a {\em system} of two equations ($i=1,2$; $j=2,1$) for the two
unkowns $s_1$ and $s_2$.
Having determined (if possible) $s_1$ and $s_2$, the recoveries $r_i$ follow immediately from \eqref{r_i_c}.
The fundamental question in this example is therefore: given $a_i\geq 0$ and $b_i\geq 0$ ($i=1,2$), has
system \eqref{s_i_c}
a solution $(s^*_1,s^*_2)$, and if so, is this solution unique? If both was the case
and if $b_i$ ($i=1,2$) was assumed to be fixed, then 
\begin{equation}
(s^*_1,s^*_2) = (s^*_1(a_1,a_2),s^*_2(a_1,a_2))
\end{equation}
would be an implicit function (in finance terms: a derivative) of the $a_i$ ($i=1,2$). 
If this function was (Lebesgue-)measurable, then one would know how to price the equities and loans 
at times before maturity using risk-neutral valuation for a given complete model of the exogenous assets.

Similarly to the above case, we could consider a set-up with no equity cross-ownership, but cross-ownership 
of debt instead: assume that company $i$ ($i=1,2$) owns 50\% of the debt of company $j$ ($j=2,1$). 
Again, assume that ownership of the debt of any of the two companies is an homogeneous asset class where all 
owners have the same rights in proportion. The new system for no-arbitrage equity prices would therefore be
($i=1,2$; $j=2,1$)
\begin{equation}
\label{s_i_cb}
s_i =  (a_i + 0.5r_j - b_i)^+ ,
\end{equation}
and the recovery of debt at maturity would be ($i=1,2$; $j=2,1$)
\begin{equation}
\label{r_i_cb}
r_i = \min\{b_i, (a_i + 0.5r_j)^+\}.
\end{equation}
The $(\cdot)^+$ in the last expression is there since debt recovery can not be negative --
similar to the equity owned by the shareholder. 
While this time the equity system \eqref{s_i_cb} is straightforward, the important question is
whether the loan recovery system \eqref{r_i_cb} has a unique solution $(r^*_1,r^*_2)$.

The answer to the questions asked in both of the above examples is {\em yes}: no-arbitrage prices
exist and they are unique. This is a direct consequence of the main result of this paper
in Section \ref{Existence and uniqueness results}.
Existence and uniqueness of no-arbitrage prices also apply in much more general set-ups than above, for
instance, we can allow cross-ownership of equities {\em and} loans (and even derivatives that have been
issued on these) at any percentage level at the same time and
for any number of companies with differing debt levels. This is a major generalization of Merton's original 
model and should have direct implications for many existing credit and derivatives pricing models in theory and 
practice that use the Merton model as a basis.

To give a simple numeric example, assume for each of the three set-ups above (no cross-ownership, 50\% stock
cross-ownership, 50\% bond cross-ownership) that the companies $1$ and $2$ have exactly the same debt
levels, $b_1=b_2$, and that they own exactly the same type of exogenous assets, i.e.~$a_1 = a_2$. 
The no-arbitrage prices for their equities and debts will therefore be identical.
Table \ref{tab:1} contains in each row for each of the three different set-ups a value (scenario) for $a_1=a_2$
that implies the same no-arbitrage prices under all three set-ups.
\begin{table}[htb]
\caption{Example for $b_1=b_2=100$}
\label{tab:1}       
\vspace{2mm}
\begin{center}
\begin{tabular}{|c||c||c||c|c|}
\hline
no XOS & 50\% stock XOS  & 50\% bond XOS & \multicolumn{2}{c|}{no-arbitrage prices}\\
$a_1=a_2$ & $a_1=a_2$ & $a_1=a_2$ & $s_1=s_2$ &  $r_1=r_2$ \\
\hline
150 & 125 & 100 & 50 & 100 \\
100 & 100 &  50 &  0 & 100 \\
 50 &  50 &  25 &  0 &  50 \\
\hline
\end{tabular}
\end{center}
\end{table}
While large-scale numerical examples and more detailed examples lie outside the scope of this paper,
a few remarks regarding Table \ref{tab:1} seem apppropriate. First, the second row clearly marks the bankruptcy 
level for all three set-ups. Second, in both cases, 50\% stock cross-ownership and 50\% bond cross-ownership,
starting with the first row, a change of -75 (or for both firms a total of -150) in the exogenous assets, down
to the level in row three, is enough to wipe out all equity and 50\% of the bond value. This is a loss of value
of 100 per firm, or 200 for both firms together. In the case of no cross-ownership, a change of -100 (or for 
both firms a total of -200) is needed for the same result. So, as should be entirely expected, starting
at the same level of equity and debt, the leverage caused by the cross-ownership structure in
the market clearly creates a higher risk regarding moves in the prices of the exogenous assets. This is 
a feature classical corporate debt models or credit risk models can not properly reflect since they are not
attempting to model cross-ownership structures that are present. Mere modelling of correlations or other
dependency structures between exogenously priced assets (which is possible within the model of this paper
as well) is not sufficient to reflect the actual interdependencies that stem from the cross-ownership of assets.


\section{Notation and mathematical preliminaries}
\label{Notation and mathematical preliminaries}

We write $\mathbf{M} > \mathbf{N}$ if for two matrices 
$\mathbf{M}, \mathbf{N} \in \R^{n\times n}$ with $\mathbf{M} = (M_{ij})_{i,j=1,\ldots,n},
\mathbf{N} = (N_{ij})_{i,j=1,\ldots,n}$ we have $M_{ij}\geq N_{ij}$
for $i,j\in\{1,\ldots,n\}$, and $M_{ij}> N_{ij}$ for at least one pair $(i,j)$.
A matrix $\mathbf{M}\in \R^{n\times n}$ is called {\em left substochastic} if $\mathbf{M}$ is
non-negative, i.e.~$M_{ij}\geq 0$ for $i,j\in\{1,\ldots,n\}$, and for any $j$ one has
$\sum_{i=1}^{n} M_{ij} \leq 1$, 
i.e.~the sums of columns are less than or equal to 1. We will call a column $j$ of $\mathbf{M}$
{\em strictly substochastic} if $\sum_{i=1}^{n} M_{ij} < 1$. We will call $\mathbf{M}$
{\em strictly left substochastic} if all columns of $\mathbf{M}$ are strictly substochastic.
A {\em right substochastic} matrix is a transposed left substochastic matrix.
A left substochastic matrix can be interpreted as an ownership matrix (see also
Ritzberger and Shorish (2002)). 
The meaning and usage of an ownership matrix will be explained in Section 
\ref{General liabilities with multiple priorities}
We will distinguish between column vectors $\mathbf{a} = (a_1, \ldots, a_n)^t \in \R^{n\times 1}$
and row vectors $\mathbf{a}^t = (a_1, \ldots, a_n)\in \R^{1\times n}$. However, for column vectors 
we will often conveniently write $\mathbf{a}\in\R^n$ and $\mathbf{a} = (a_1, \ldots, a_n)$. 
From the context, no confusion should arise.
The meaning of $\mathbf{a} > \mathbf{b}$
and $\mathbf{a}^t > \mathbf{b}^t$ for $\mathbf{a}, \mathbf{b}\in \R^n$ 
is analogue to the conventions for matrices. We use the symbols
$\mathbf{0} = (0,\ldots,0)$ (where $\mathbf{0}$ might also be used for a zero
matrix), $\mathbf{1} = (1,\ldots,1)$, and $\mathbf{I}$ for the identity matrix.
In all cases where these symbols are used, the dimension should be clear from the context.
The following operations will apply element-wise to matrices and vectors:
the positive part, $(\cdot)^+$; the negative part, $(\cdot)^-$; the maximum, $\max\{\cdot,\cdot\}$;
the minimum, $\min\{\cdot,\cdot\}$. 
We will make use of the $\ell^1$-norm on $\R^n$, for $\mathbf{x}\in\R^n$ given by
\begin{equation}
||\mathbf{x}||_1 = \sum_{i=1}^{n} |x_i| .
\end{equation}
For the $\ell^1$-norm of any strictly left substochastic matrix $\mathbf{M}$ we have ($\mathbf{x}\in\R^n$)
\begin{equation}
\label{matrix-norm}
||\mathbf{M}||_1 \; = \; \max_{||\mathbf{x}||_1=1}  ||\mathbf{M}\cdot\mathbf{x}||_1 \; = \;
\max_{j} \sum_{i=1}^{n} M_{ij} \; < \; 1 ,
\end{equation}
that means $||\mathbf{M}||_1$ is the maximum of the column sums. As expected, it follows that
$||\mathbf{M}\cdot\mathbf{x}||_1 \leq ||\mathbf{M}||_1 \cdot ||\mathbf{x}||_1$, since
\begin{align}
\label{l_1}
||\mathbf{M}\cdot\mathbf{x}||_1 & \; = \; \sum_{i=1}^{n} \left|\sum_{j=1}^{n} M_{ij}x_j \right| 
\; \leq \; \sum_{j=1}^{n}|x_j|\cdot\sum_{i=1}^{n} M_{ij}\\
\nonumber & \; \leq \; \left(\max_{j} \sum_{i=1}^{n} M_{ij}\right)\cdot ||\mathbf{x}||_1 
\; = \; ||\mathbf{M}||_1 \cdot ||\mathbf{x}||_1 .
\end{align}


\section{No-arbitrage pricing of general liabilities}
\label{Equilibrium for liabilities with multiple priorities}


\subsection{General liabilities of differing seniority}
\label{General liabilities with multiple priorities}

In Section \ref{Cross-ownership: a two-firm example} we considered a situation where only stocks and bonds 
could be cross-owned and where the priorities of claims in a possible liquidation were clear by default
(bonds more senior than stocks). In this section, where we consider the case with $n$ companies
($n\in\{1,2,\ldots\}$), we allow the cross-ownership of equities and more general liabilities, including debts
{\em and} derivatives. For these liabilities we allow differing priorities of the corresponding
recovery claims under liquidation.

Suppose that the vector
$\mathbf{a} = (a_i)_ {i=1,\ldots,n} \geq \mathbf{0}$ summarily denotes the exogenous assets of
market value $a_i\geq 0$ held by company $i$ ($i=1,\ldots,n$) at maturity.
These assets are exogenously priced in the sense that it is assumed that the capital structure 
of the companies $1,\ldots,n$ has no influence on the price of these assets.
The exogenous assets $\mathbf{a}$ could include physical assets like commodities or property,
the work force of a company, intellectual property, but also cash, future cash flows 
or claims on (equities or liabilities of) external firms as long as these are not affected by the 
considered $n$ firms' capital structures. We assume that the exogenous
assets or parts thereof can be sold at any time at the given price $\mathbf{a}$. A sale
of only a part of the assets would not affect the price of the remaining parts.
We will see in Section \ref{Extended exogenous assets} that the dimension of the exogenous assets'
price vector ($n$ above, and hence identical to the number of firms) is irrelevant to our considerations, and
any number of exogenous assets could be considered. We choose the dimension $n$ for convenience only.

As mentioned earlier, a left substochastic matrix can be interpreted as an {\em ownership matrix}.
Let $\mathbf{s}\in\R^n$ denote the prices of the equities of the companies $1,\ldots,n$. 
Furthermore, suppose that the equities of the companies $1,\ldots,n$ with
prices $\mathbf{s}$ are at least partially owned by the companies themselves. In particular, we assume
that company $i$ owns a proportion $0 \leq M^s_{ij} \leq 1$ of the equity of company $j$
(called the ``cross-owned fraction" in B{\o}hren and Michalsen (1994)).
This (partial) ownership is worth $M^s_{ij}s_j$. The ownership structure of the equities can therefore be described
by the left substochastic matrix $\mathbf{M}^s = (M^s_{ij})_{i,j=1,\ldots,n}$ (cf.~Ritzberger and Shorish (2002)).
The total value of any equity that company $i$ owns is given by the $i$-th entry of the vector
$\mathbf{M}^s\cdot\mathbf{s}$, i.e.~by
\begin{equation}
\sum_{j=1}^{n}M^s_{ij}s_j  . 
\end{equation}
Similarly, we can describe the cross-ownership structure of
any outstanding liabilities by means of a left substochastic matrix.

In order to properly reflect cross-ownership structures that are present in a group of firms, 
any firm that owns equities or liabilities of firms in this group, where at the same time the group owns
part of the equity or liabilities of the considered firm, should be included in the model. Firms which only
own equities or liabilities of other firms, but who's own equity or liabilities are not at least partially
owned by other firms, do not have to be considered since they are price takers in this market in the sense
that their own assets and liabilies do not influence other firms' balance sheets.

\begin{assumption}
\label{master-def_1}
Let $\mathbf{M}^s$ and $\mathbf{M}^{d,i}$ $(i=1,\ldots,m)$ be strictly left substochastic matrices. 
Let $\mathbf{a}\in (\R_0^+)^n$. For $i=1,\ldots,m$, let $\mathbf{d}^i$ be a function
\begin{eqnarray}
\R^{n(m+1)} & \longrightarrow & (\R_0^+)^n \\
\left(
\begin{array}{*{3}{l}}
\mathbf{r}^1\\
\vdots \\
\mathbf{r}^m\\
\mathbf{s}
\end{array}
\right) 
& \longmapsto & \mathbf{d}^i_{\mathbf{r}^1,\ldots,\mathbf{r}^m,\mathbf{s}} .
\end{eqnarray}
\end{assumption}

Suppose now that the $n$ elements of the vector $\mathbf{d}^i$ define the liabilities of the $n$ considered
firms, where $i=1,\ldots,m$ are the priorities of the corresponding claims in a possible liquidation of
any firm (e.g.~by Chapter 7), such that 1 is the highest priority (paid first) and $m$ is the
lowest (paid last). We assume that these liabilities are payable in cash.
The functions $\mathbf{d}^i$ are non-negative since negative liabilities are assets
(another firm's liability) and hence they will be modelled as such.
Suppose that the $\mathbf{r}^i \in \mathbb{R}^n$ 
are the vectors of recovery claims belonging to $\mathbf{d}^i$. So, while
the $\mathbf{d}^i$ describe what is supposed to be paid at maturity, the $\mathbf{r}^i$ stand for
the actual payoff, which could be after a liquidation due to the default of the writer of the liability.
In general, we should therefore have 
\begin{equation}
\mathbf{0} \leq \mathbf{r}^i \leq \mathbf{d}^i .
\end{equation}

For simplicity, we assume that there is only one liability per firm per level of seniority. 
However, as long as it was clear how to split the corresponding recovery claim in any economic scenario
of a liquidation situation in which the claim would be recovered at less than par, it would be no problem
to assume that any of these liabilities was a sum of several liabilities of the same seniority. If such
an agreement (of a split) did not exist, one could possibly assume a repayment {\em pari passu}.

A very simple version of a liabilities function $\mathbf{d}^i$ 
would be where $\mathbf{d}^i_{\mathbf{r}^1,\ldots,\mathbf{r}^m,\mathbf{s}} \equiv \mathbf{b}^i\in (\R^+_0)^n$. 
In this case, the liabilities would be simple loans, rather than general liabilities, like derivatives
that could depend on other assets like $\mathbf{s}$. 
However, since $\mathbf{d}^i_{\mathbf{r}^1,\ldots,\mathbf{r}^m,\mathbf{s}}$ has the $\mathbf{r}^i$ as 
arguments, the liability functions
can even depend on their own(!) eventual payoff (see an example in Sec.~\ref{No unique equilibrium}).
Assumption \ref{master-def_1} is also very general in the sense that {\em no restrictions on derivatives
regarding exogenous assets} have been made. This is because exogenous assets are treated as constants 
in our framework since all considerations are conditional on given prices for exogenous
assets. So, to be very clear about this and although we will not use this notation later, the
$\mathbf{d}^i$ may also depend on the exogenously priced assets $\mathbf{a}$, i.e.~one generally has
\begin{equation}
\label{liabilities a}
\mathbf{d}^i_{\mathbf{r}^1,\ldots,\mathbf{r}^m,\mathbf{s}} = 
\mathbf{d}^i_{\mathbf{r}^1,\ldots,\mathbf{r}^m,\mathbf{s},\mathbf{a}} .
\end{equation}
For an example see Section \ref{Example with stocks, bonds, and derivatives}.

Regarding the issue of strictly substochastic matrices in Assumption \ref{master-def_1} (rather than just
substochastic ones), an example in Sec.~\ref{No equilibrium under maximum cross-ownership} will illustrate 
the kind of problems that can arise if ownership matrices are not strictly substochastic. 
A more thorough discussion of the possible effects of non-strictly substochastic matrices 
is beyond the scope of this paper.


\subsection{The liquidation value equations}
\label{The equilibrium equations}

Suppose now that the matrices $\mathbf{M}^s$ and $\mathbf{M}^{d,i}$ $(i=1,\ldots,m)$ are 
ownership matrices as in Section \ref{General liabilities with multiple priorities}. 
Assume that $\mathbf{M}^s$ describes the equity
(common stock) cross-ownership in the system of $n$ firms, while the $\mathbf{M}^{d,i}$ $(i=1,\ldots,m)$
describe the cross-ownership of general liabilities, with $\mathbf{M}^{d,i}$ belonging to $\mathbf{d}^i$. 
Under this set-up and under Assumption \ref{master-def_1}, we demand the following at maturity:

\begin{assumption}[Absolute Priority Rule]
\label{equilibrium assumption 1}
The priority of claims is honored. Equity is the residual claim.
\end{assumption}

The Absolute Priority Rule, which requests strict adherence to the seniority of liabilities
in a liquidation in the sense that any higher rank claim has to be fully paid off before any lower
rank claim can be paid, is not always honored in real life. Longhofer and  Carlstrom (1995)
write:
\begin{quote}
While this rule would seem quite simple to implement, it is routinely circumvented in practice. 
In fact, bankruptcy courts themselves play a major role in abrogating this feature of debt contracts.
\end{quote}
However, the circumstances (liquidation vs.~reorganization, or political intervention) and also the
extent to which the rule is disregarded vary (see e.g.~Eberhart, Moore and Roenfeldt (1990)) 
and can not be discussed in this paper. For the purpose of our theory, Assumption 
\ref{equilibrium assumption 1} is reasonably close to practice.

Under Assumption \ref{equilibrium assumption 1}, the equations 
\begin{eqnarray}
\label{master_n}
\mathbf{r}^1 & = & \min\left\{\mathbf{d}^1_{\mathbf{r}^1,\ldots,\mathbf{r}^m,\mathbf{s}},
\quad\, \mathbf{a} + \mathbf{M}^s\cdot\mathbf{s}  
+ \sum_{i=1}^{m} \mathbf{M}^{d,i}\cdot\mathbf{r}^i \right\}\\
\nonumber & & \mbox{For $0 < j < m$:} \\
\label{master_j} 
\mathbf{r}^{j+1} & = & \min\left\{\mathbf{d}^{j+1}_{\mathbf{r}^1,\ldots,\mathbf{r}^m,\mathbf{s}},\;
\left(\mathbf{a} + \mathbf{M}^s\cdot\mathbf{s}  
+ \sum_{i=1}^{m} \mathbf{M}^{d,i}\cdot\mathbf{r}^i 
- \sum_{i=1}^{j} \mathbf{d}^i_{\mathbf{r}^1,\ldots,\mathbf{r}^m,\mathbf{s}} \right)^+ \right\} \\
\label{master_1}
\mathbf{s} & = & \qquad\qquad\qquad\quad\; \left(\mathbf{a} + \mathbf{M}^s\cdot\mathbf{s}  
+ \sum_{i=1}^{m} \mathbf{M}^{d,i}\cdot\mathbf{r}^i 
- \sum_{i=1}^{m} \mathbf{d}^i_{\mathbf{r}^1,\ldots,\mathbf{r}^m,\mathbf{s}} \right)^+ 
\end{eqnarray}
follow for the {\em liquidation values} of $\mathbf{r}^1,\ldots,\mathbf{r}^m$ and $\mathbf{s}$.
Liquidation means here that at maturity all assets are converted into cash and subsequently all
equity is paid out in cash as well. 
Because of Assumption \ref{equilibrium assumption 1}, {\em each claim on a liability will only pay the
minimum of the promised payoff and the value of the remaining assets of the firm after liabilities
of higher seniority have been paid.}
In this sense, the equations \eqref{master_n} -- \eqref{master_1} are the generalization of the 
Merton equations \eqref{s_i} and \eqref{r_i} for the multi-firm case with general liabilities under
cross-ownership. However, while in Merton's case the liquidation of equity is trivial since any
equity is identical to the remaining exogenous assets after paying the debt, it is much
less obvious how such a liquidation would work in our case due to the equity entanglement caused by
cross-ownership. For a further and more thorough discussion of the issue of liquidation see
therefore Section \ref{A comment on solvency, liquidation and arbitrage}. In that section we will also
show under fairly reasonable assumptions that the equations \eqref{master_n} -- \eqref{master_1} 
must hold for the no-arbitrage prices of equities and liabilities at maturity. Hence, no-arbitrage
prices and liquidation values are identical.

Note that $\mathbf{r}^1$ are the recovery claims with the highest priority
in a liquidation, and $\mathbf{r}^m$ are the recoveries of the lowest priority claims. 
Equity $\mathbf{s}$ is, of course, the first asset to be wiped out, then, in this order, 
the recoveries $\mathbf{r}^m$ to $\mathbf{r}^1$ get wiped out (component-wise). The model is
flexible enough to not only incorporate bonds and derivatives of differing seniority,
but also some derivatives could rank higher than some bonds.
We call a non-negative solution $(\mathbf{r}^{1*},\ldots,\mathbf{r}^{m*},\mathbf{s}^*)$
of the system \eqref{master_n} -- \eqref{master_1} a {\em no-arbitrage price equilibrium}.
In Assumption \ref{equilibrium assumption 1} we defined equity as the ``residual
claim". In practice, the asset class closest to this would be ``common stock". The asset class of
``preferred stock", which usually has more rights and a higher seniority in a default event,
would in our set-up be modelled as one of the liabilities ranking higher than equity. 

In a first attempt, one would possibly formulate the liquidation value equations 
\eqref{master_n} -- \eqref{master_1} such that the first one would read
\begin{equation}
\label{master_n_alt}
\mathbf{r}^1 \; = \; \min\left\{\mathbf{d}^1_{\mathbf{r}^1,\ldots,\mathbf{r}^m,\mathbf{s}},
\; \left(\mathbf{a} + \mathbf{M}^s\cdot\mathbf{s}  
+ \sum_{i=1}^{m} \mathbf{M}^{d,i}\cdot\mathbf{r}^i \right)^+\right\} ,
\end{equation}
which is in line with the fact that also the highest priority recipient of liquidation proceeds
can never be forced to pay (rather than receive) during the liquidation process. The lemma below
shows that the system \eqref{master_n} -- \eqref{master_1} is equivalent to the possibly more
natural approach given by the system \eqref{master_n_alt}, \eqref{master_j} and \eqref{master_1}.
We will prefer to work with \eqref{master_n} -- \eqref{master_1} since it is somewhat leaner 
and simplifies some of our proofs.

\begin{lemma}
\label{master-pos_sols}
Under Assumption \ref{master-def_1}, any solution $(\mathbf{r}^{1*},\ldots,\mathbf{r}^{m*},\mathbf{s}^*)$
of the system \eqref{master_n} -- \eqref{master_1} is non-negative. Hence,  \eqref{master_n} -- 
\eqref{master_1} and \eqref{master_n_alt}, \eqref{master_j} and \eqref{master_1} have identical solutions.
\end{lemma}
For the proof of the lemma we will use the following notation.
Let $\pi$ be a permutation (i.e.~a bijection)
on $\{1,\ldots,n\}$ and $\mathbf{a}\in\R^{n}$ and $\mathbf{M}\in\R^{n\times n}$. 
We denote ${_{\pi}\mathbf{a}} = (a_{\pi(1)}, \ldots, a_{\pi(n)})$, and we denote
${_{\pi\pi}\mathbf{M}}$ for the matrix obtained from $\mathbf{M}$ by permuting elements 
such that ${_{\pi\pi}M_{ij}}= M_{\pi(i)\pi(j)}$ for $i,j\in\{1,\ldots,n\}$. The latter
is a simultaneous permutation of rows and columns.

\begin{proof}
Let $(\mathbf{r}^{1*},\ldots,\mathbf{r}^{m*},\mathbf{s}^*)$ be a solution of 
\eqref{master_n} -- \eqref{master_1}. Hence,
$\mathbf{r}^{2*},\ldots,\mathbf{r}^{m*},\mathbf{s}^*\in(\R_0^+)^n$.
Consider now Eq.~\eqref{master_n}, and substitute with a solution $(\mathbf{r}^{1*},\ldots,\mathbf{r}^{m*},\mathbf{s}^*)$
which contains at least one $r^{1*}_i < 0$ ($i\in\{1,\ldots,n\}$). 
By the help of a permutation $\pi$ on $\{1,\ldots,n\}$, we can now re-arrange rows and columns
of this new system such that it has the equivalent form 
\begin{equation}
{_{\pi}\mathbf{r}}^{1*} \; = \;
\min\left\{{_{\pi}(\mathbf{d}^1_{\mathbf{r}^{1*},\ldots,\mathbf{r}^{m*},\mathbf{s}^*})},
\; {_{\pi}\mathbf{a}} + {_{\pi\pi}\mathbf{M}}^s\cdot {_{\pi}\mathbf{s}}^*  
+ \sum_{i=1}^{m} {_{\pi\pi}\mathbf{M}}^{d,i}\cdot {_{\pi}\mathbf{r}}^{i*} \right\} ,
\end{equation}
where rows 1 to $j$ ($j\in\{1,\ldots,n\}$) are non-negative, and rows $j+1$ to $n$ are negative. 
The matrices ${_{\pi\pi}\mathbf{M}}^s$ and ${_{\pi\pi}\mathbf{M}}^{d,i}$
are again strictly left substochastic matrices since 
sums of columns are still less than 1. Consider now the subsystem that
consists of the negative rows from $j+1$ to $n$. Row $k$ ($j+1\leq k \leq n$) has the form
\begin{equation}
{_{\pi}r}^{1*}_k  = {_{\pi}a}_k + \sum_{l=1}^{n} {_{\pi\pi}M^s_{kl}}\cdot{_{\pi}s}^*_{l} 
+ \sum_{i=1}^{m}\sum_{l=1}^{n} {_{\pi\pi}M^{d,i}_{kl}}\cdot{_{\pi}r}^{i*}_{l} .
\end{equation}
However, since 
${_{\pi}\mathbf{r}}^{2*},\ldots,{_{\pi}\mathbf{r}^{m*}},{_{\pi}\mathbf{s}}^*\in(\R_0^+)^n$
and ${_{\pi}r}^{1*}_{l} \geq 0$ for $1\leq l \leq j$, we can write
\begin{equation}
\label{n-j-sysm}
{_{\pi}r}^{1*}_k  = c_k + \sum_{l=j+1}^{n} {_{\pi\pi}M^{d,1}_{kl}}\cdot{_{\pi}r}^{1*}_{l} 
\end{equation}
for some $c_k \geq 0$. Define now 
\begin{eqnarray}
\mathbf{r}' & = & ({_{\pi}r}^{1*}_{j+1},\ldots,{_{\pi}r}^{1*}_{n})^t , \\
\mathbf{c} & = & (c_{j+1},\ldots,c_{n})^t , \\
\mathbf{M}' & = & ({_{\pi\pi}M^{d,1}_{kl}})_{k,l=j+1,\ldots,n} . 
\end{eqnarray}
The system \eqref{n-j-sysm} (for $j+1\leq k \leq n$) can now be written as
\begin{equation}
\label{n-j-sys2m}
\mathbf{r}' = \mathbf{c} + \mathbf{M}' \cdot \mathbf{r'} ,
\end{equation}
where $\mathbf{r'} < \mathbf{0}$ and $\mathbf{c} \geq \mathbf{0}$. Since the matrix $\mathbf{M}'$
as the lower right $(n-j)\times (n-j)$-submatrix of ${_{\pi\pi}\mathbf{M}^{d,1}}$ is again a
strictly left substochastic matrix, we can apply Lemma \ref{M-lemma} of Section 
\ref{A result for substochastic matrices} to the system
\begin{equation}
\label{n-j-sys3m}
\mathbf{r'} = (\mathbf{I} - \mathbf{M}')^{-1} \cdot \mathbf{c} ,
\end{equation}
which is equivalent to \eqref{n-j-sys2m}. 
Since $(\mathbf{I} - \mathbf{M}')^{-1}$ is non-negative
according to Lemma \ref{M-lemma}, Eq.~\eqref{n-j-sys3m} is a contradiction.
\end{proof}


\subsection{Accounting equations}
\label{The balance sheet equation}

Before we turn to results about the existence of no-arbitrage price equilibria in Section 
\ref{Existence and uniqueness results},
it is useful to first consider the balance sheet equations under no-arbitrage, as well as measures of
leverage and cross-ownership (in Sec.~\ref{Measures of leverage and cross-ownership}). 

Under no-arbitrage, the accounting equations (or balance sheet equations),
\begin{eqnarray}
\label{balance sheet}
\mathbf{a} + \mathbf{M}^s\cdot\mathbf{s}  
+ \sum_{i=1}^{m} \mathbf{M}^{d,i}\cdot\mathbf{r}^i
& = & 
\mathbf{s} + \sum_{i=1}^{m} \mathbf{r}^i \\
\nonumber \text{assets $+$ receivables} & = & \text{equity $+$ payable liabilities} ,
\end{eqnarray}
follow directly from applying Eq.~\eqref{capcons} in Lemma \ref{remark1} in Section 
\ref{A result for the sum of certain differences}
to the right side of the sum of the system \eqref{master_n} -- \eqref{master_1}. 
Under no-arbitrage, Eq.~\eqref{balance sheet} holds componentwise (per firm) with all components being
non-negative.
It therefore also holds in absolute terms ($\ell_1$-terms) for the total capital in the considered 
system:
\begin{eqnarray}
\label{total_conservation_2}
\left|\left| \mathbf{a}\right|\right|_1 + \left|\left|\mathbf{M}^s\cdot\mathbf{s}\right|\right|_1  
+ \sum_{i=1}^{m} \left|\left|\mathbf{M}^{d,i}\cdot\mathbf{r}^i \right|\right|_1
& = & 
\left|\left| \mathbf{s}\right|\right|_1 + \sum_{i=1}^{m} \left|\left|\mathbf{r}^i \right|\right|_1 \\
\nonumber \text{total assets $+$ total receivables} & = & \text{total equity $+$ total payable liabilities}.
\end{eqnarray}
The right side of \eqref{total_conservation_2} represents all claims on equity and liabilities
of the financial firms $1,\ldots,n$. The left side, apart from $\left|\left| \mathbf{a}\right|\right|_1$,
represents all claims on equity and liabilities that are claimed {\em within} the system of the firms
$1,\ldots,n$ itself. Expressed differently, the following mathematically trivial 
(from \eqref{total_conservation_2}) but from an economic perspective certainly important
equation,
\begin{eqnarray}
\label{Modigliani-Miller_eq}
\underbrace{\left|\left| \mathbf{s}\right|\right|_1 
+ \sum_{i=1}^{m} \left|\left|\mathbf{r}^i \right|\right|_1} 
\; - \;
\left( \underbrace{ \left|\left|\mathbf{M}^s\cdot\mathbf{s}\right|\right|_1  
+ \sum_{i=1}^{m} \left|\left|\mathbf{M}^{d,i}\cdot\mathbf{r}^i \right|\right|_1 } \right)
& = & 
\left|\left| \mathbf{a}\right|\right|_1  \\
\nonumber \text{all claims} \qquad - \qquad\quad \text{internally held claims} \qquad
& = & \text{exogenous assets},  
\end{eqnarray}
means that under no-arbitrage the sum of the values of {\em all externally held claims} regarding the 
financial firms $1,\ldots,n$ is identical to the value of {\em all exogenous} assets owned by these firms. 
It is remarkable that ownership structures and the amount of financial leverage therefore have no
influence on the aggregate value of externally held assets regarding these companies. 
In other terms, all asset value that stems from internal leverage (internal leverage will be properly defined
in Eq.~\eqref{L} below) is contained within the system and hence somewhat irrelevant to outsiders who hold
claims of this system but who are not themselves partially owned or financially liable to this system. 
Although being relevant for a system of firms only, and not for individual firms, 
this seems to be an interesting {\em principle of capital structure irrelevance}. 
We therefore state it as a separate theorem. 

\begin{theorem}
\label{CSI}
Under the Assumptions \ref{master-def_1} and \ref{equilibrium assumption 1},
the no-arbitrage value of all externally held claims (liabilities and equity) belonging to a group 
of firms is identical to the no-arbitrage value of the exogenous assets owned by the group:
\begin{equation}
\label{MM}
\text{value of externally held claims} = \text{value of exogenous assets}.
\end{equation}
The aggregate value of the externally held claims
is therefore independent of any other aspects of the financial structure of this group.
\end{theorem}


\subsection{Measures of leverage and cross-ownership}
\label{Measures of leverage and cross-ownership}

From Eq.~\eqref{Modigliani-Miller_eq} in the previous section it is fairly obvious that a high
degree of cross-ownership within a group of financial firms (i.e.~when the
matrices $\mathbf{M}^{\cdot}$ have large entries; for instance, all column sums could be close to 1)
leads to artificially large balance sheets when compared with the
underlying exogenous asset base. In this section we therefore look at
measures of financial leverage and cross-ownership under no-arbitrage.

It is clear that the standard measure of financial leverage in the case of a
single firm, the debt-to-equity ratio (it should be {\em liabilities}-to-equity anyway), 
is not very useful if we consider a
whole group of firms and if we want to assess the degree of their
(internal) leverage as a group. In fact, Eq.~\eqref{Modigliani-Miller_eq} shows
that even in the case of no liabilities at all, the cross-ownership of
equity alone could cause balances to be a
large multiple of the value of the underlying exogenous assets. We therefore
need to find a meaningful measure of the {\em internal} leverage for a group
of financial firms.
Assume now $\left|\left| \mathbf{a}\right|\right|_1 > 0$, and define 
\begin{eqnarray}
\label{L}
L & = & \frac{\left|\left|\mathbf{M}^s\cdot\mathbf{s}\right|\right|_1  
+ \sum_{i=1}^{m} \left|\left|\mathbf{M}^{d,i}\cdot\mathbf{r}^i \right|\right|_1}
{\left|\left|\mathbf{a} \right|\right|_1} ,\\
& = &  \frac{\text{internally held claims}}{\text{exogenous assets}} \; 
\stackrel{\eqref{MM}}{=} \;  
\frac{\text{internally held claims}}{\text{externally held claims}} .
\end{eqnarray}
This value can be interpreted as the level of {\em internal financial leverage} in the 
considered system. Note that $L=L(\mathbf{a})$ since not only $||\mathbf{a}||_1$, but
also $\mathbf{s}$ and the $\mathbf{r}^i$ depend on $\mathbf{a}$.
Furthermore, by \eqref{total_conservation_2},
\begin{eqnarray}
\label{Lb}
L + 1 & = &  
\frac{\left|\left| \mathbf{s}\right|\right|_1 + \sum_{i=1}^{m} \left|\left|\mathbf{r}^i \right|\right|_1}
{\left|\left|\mathbf{a} \right|\right|_1} ,\\
\label{Lbb}
& = & \frac{\text{total claims}}{\text{externally held claims}} 
\; = \; \frac{\text{total assets}}{\text{exogenous assets}} .
\end{eqnarray}
Related to \eqref{L} is for 
$\left|\left| \mathbf{s}\right|\right|_1 + \sum_{i=1}^{m} \left|\left|\mathbf{r}^i \right|\right|_1 > 0$
the value defined by (cf.~\eqref{Modigliani-Miller_eq})
\begin{eqnarray}
\label{I}
I & = & \frac{\left|\left|\mathbf{M}^s\cdot\mathbf{s}\right|\right|_1  
+ \sum_{i=1}^{m} \left|\left|\mathbf{M}^{d,i}\cdot\mathbf{r}^i \right|\right|_1}
{\left|\left| \mathbf{s}\right|\right|_1 + \sum_{i=1}^{m} \left|\left|\mathbf{r}^i \right|\right|_1} 
\; = \; \frac{L}{L+1}  ,\\
\label{MMLast} & = &  \frac{\text{internally held claims}}{\text{total claims}} ,
\end{eqnarray}
which could be seen as a measure of the {\em degree of financial cross-ownership}. We could also
consider this to be a measure of financial inbreeding or self-excitement.
Because of $||\mathbf{M}^s||_1 , ||\mathbf{M}^{d,i}||_1 \in [0,1)$ for 
$i=1,\ldots,m$ (cf.~Eq.~\eqref{matrix-norm}), it follows from \eqref{l_1} that for
\begin{equation}
\label{I^max}
I^{\max} = \max\{|| \mathbf{M}^s||_1, ||\mathbf{M}^{d,1}||_1,\ldots,||\mathbf{M}^{d,m}||_1 \} 
\end{equation}
one has 
\begin{equation}
\label{I<=I^max}
0 \; \leq \; I \; \leq \; I^{\max} \; < \; 1 .
\end{equation}
Since $L=\frac{I}{1-I}$ and since $L$ is a strictly monotonically increasing function of $I$, 
it follows immediately that for 
\begin{equation}
\label{L^max}
L^{\max}=\frac{I^{\max}}{1-I^{\max}}
\end{equation}
one has 
\begin{equation}
\label{L<=L^max}
0 \; \leq \; L \; \leq \; L^{\max} \; < \; +\infty 
\end{equation}
and 
\begin{equation}
\label{upper_balance_boundary}
\left|\left| \mathbf{s}\right|\right|_1 + \sum_{i=1}^{m} \left|\left|\mathbf{r}^i \right|\right|_1
\; \stackrel{\eqref{Lb}}{=} \; (L+1)||\mathbf{a}||_1 \; \leq \; (L^{\max}+1)||\mathbf{a}||_1 .
\end{equation}
The upper boundary $L^{\max}$ is sharp, as an example in Section \ref{L_max is sharp} demonstrates.
Another straightforward conclusion from the Equations \eqref{I<=I^max}
and \eqref{total_conservation_2} is that $||\mathbf{a}||_1 = 0$ is equivalent to 
$||\mathbf{s}||_1 = ||\mathbf{r}^1||_1 = \ldots = ||\mathbf{r}^m||_1 = 0$.

The value $L+1$ is the value of the sum of the balance sheets of all firms expressed in 
terms of the sum of all exogenous assets (cf.~Eq.~\eqref{Lbb}). 
It seems therefore plausible that the higher $L$, or, equivalently, the lower 
the percentage of exogenous assets in the balance sheets, the higher the risk for the 
balance sheets that stems from these exogenous assets in absolute terms. Such risk could for instance
materialize in the form of an instantaneous shock in the prices of the exogenous assets just before
(at) maturity. The values $L$ and $I$ could therefore also be seen as measures of 
{\em systemic risk} or {\em financial contagion} in the absence of liquidity risk
(see also Section \ref{A comment on solvency, liquidation and arbitrage}).  
As we pointed out earlier, $L = L(\mathbf{a})$ and $I = I(\mathbf{a})$ (see also the example
in Table \ref{tab:2}). Therefore, the considered
measures are calculated for one particular scenario for the exogenous assets. For a physical
probability measure, $\mathbb{P}$, or for an equivalent martingale measure, $\mathbb{Q}$, for instance
taken from a stochastic model for $\mathbf{a}$ (cf.~Sec.~\ref{Equilibria for non-zero maturities}), 
one method to obtain such measures for all scenarios at the same time could be to calculate the
expectations
\begin{equation}
\label{measure before maturity}
\E_{\mathbb{P}}[L],\; \E_{\mathbb{Q}}[L] \;\leq\; L^{\max} , \quad\text{or}\quad
\E_{\mathbb{P}}[I],\; \E_{\mathbb{Q}}[I], \;\leq\; I^{\max} .
\end{equation}
It is clear that no-arbitrage prices before
maturity, obtained for instance by risk-neutral pricing (cf.~Sec.~\ref{Equilibria for non-zero maturities}),
would already reflect any risk from systemic cross-ownership and leverage.
See Sec.~\ref{Equilibria for non-zero maturities} also for a further idea how measures
of cross-ownership could be obtained at times before maturity.

In summary it can be said that {\em under no-arbitrage cross-ownership does not overstate equity} 
-- the equity is there, according to no-arbitrage prices. However, cross-ownership can somewhat 
artificially inflate balance sheets (and potentially share prices) in comparison to the
underlying exogenous assets. Cross-ownership should therefore be an integral part of any measure of 
market leverage. For instance, Table \ref{tab:2} shows the values of $L$ for the
example in Section \ref{Cross-ownership: a two-firm example} (Table \ref{tab:1}). The
first row of values demonstrates from left to right how decreasing exogenous assets prices but 
increasing leverage lead to the same no-arbitrage prices of equities and liabilities, i.e.~the
total assets on the balance sheets in the system are identical in all three cases.
\begin{table}[htb]
\caption{Example of Sec.~\ref{Cross-ownership: a two-firm example} for $b_1=b_2=100$}
\label{tab:2}       
\vspace{2mm}
\begin{center}
\begin{tabular}{|c|c||c|c||c|c||c|c|}
\hline
\multicolumn{2}{|c||}{no XOS} & \multicolumn{2}{c||}{50\% stock XOS}  & \multicolumn{2}{c||}{50\% bond XOS} & \multicolumn{2}{c|}{no-arbitrage prices}\\
$a_1=a_2$ & $L$ & $a_1=a_2$ & $L$ & $a_1=a_2$ & $L$ & $s_1=s_2$ &  $r_1=r_2$ \\
\hline
150 & 0 & 125 & 0.2 & 100 & 0.5  & 50 & 100 \\
100 & 0 & 100 & 0 &  50 & 1 &  0 & 100 \\
 50 & 0 &  50 & 0 &  25 & 1 &  0 &  50 \\
\hline
\end{tabular}
\end{center}
\end{table}
In Section \ref{Preliminaries}, we quoted Ritzberger and Shorish (2002) 
with the statement ``that under cross-ownership the book value of a firm will 
tend to be overestimated with respect to the underlying cash flows." This 
would be exactly the case here if
we considered discounted cashflows for the exogenous assets. In that sense, cross-ownership
can cause leverage that not only increases the risk of contagion, but there also exists a moral
issue as far as some investors might not be aware of the kind of {\em hidden leverage} cross-ownership
of equity can cause. However, as we have seen in Theorem \ref{CSI}, 
{\em external owners of claims and equities are in aggregate not affected by
cross-ownership leverage.}

In contrast to our measure of internal leverage, for $||\mathbf{s}||_1 > 0$, a measure for
the {\em external financial leverage} in the system could be defined as
the ratio of payable externally held liabilities to externally held equity
(cf.~{\em debt-to-equity}),
\begin{eqnarray}
\label{L_ex}
L_\text{ex} & = & \frac{
\sum_{i=1}^{m} \left|\left|\mathbf{r}^i \right|\right|_1 
-
\sum_{i=1}^{m} \left|\left|\mathbf{M}^{d,i}\cdot\mathbf{r}^i \right|\right|_1
}{
\left|\left| \mathbf{s}\right|\right|_1 - \left|\left|\mathbf{M}^s\cdot\mathbf{s}\right|\right|_1
} \\
& = & \frac{\text{externally held liabilities}}{\text{externally held equity}} .
\end{eqnarray}
However, in the light of Theorem \ref{CSI}, for external claimholders, external leverage 
is essentially as irrelevant as internal one since there is no influence on the value of all externally
held assets, which is constant $\left|\left| \mathbf{a}\right|\right|_1$. 

While we are not considering measures for the leverage of individual companies here, they certainly
serve an important purpose. However, on a market-wide scope, internal leverage as defined in
\eqref{L} is possibly the most relevant measure of leverage, even though it matters only within a
group of firms, and not outside.


\subsection{Existence and uniqueness results}
\label{Existence and uniqueness results}

\begin{assumption}
\label{assumption_1}
The functions $\mathbf{d}^i$ of Assumption \ref{master-def_1} are continuous for $i=1,\ldots,m$.
\end{assumption}

\begin{assumption}
\label{master-def_2}
For $i=1,\ldots,m$ and $j=1,\ldots,n$,
\begin{equation}
\label{master-l-def}
\mathbf{d}^i_{\mathbf{r}^1,\ldots,\mathbf{r}^m,\mathbf{s}} = 
\left(
\begin{array}{*{3}{l}}
\psi^i_j\left(
\sum_{k=1}^{n} M^s_{jk}s_k  + \sum_{l=1}^{m} \sum_{k=1}^{n} M^{d,l}_{jk}r^{l}_k 
\right)
\end{array}
\right)_{j=1,\ldots,n}
\end{equation}
where $\psi^i_j: \R \rightarrow \R^+_0$ are monotonically increasing functions such that for 
any $\mathbf{y}^1,\mathbf{y}^2\in\R^n$ with $\mathbf{y}^1 \geq \mathbf{y}^2$
\begin{equation}
\label{master-l-cond}
\mathbf{y}^1-\mathbf{y}^2 \; \geq \; \sum_{i=1}^{m} 
\left(\psi^i_j(y^1_j) - \psi^i_j(y^2_j)\right)_{j=1,\ldots,n} .
\end{equation}
\end{assumption}
Assumption \ref{master-def_2} is obviously much stronger than Assumption \ref{assumption_1}.
Condition \eqref{master-l-cond} alone is a stronger condition than Lipschitz continuity. 
Note that while Assumption \ref{master-def_2} restricts allowed liabilities (derivatives) 
in the sense that they can only be written 
honoring restriction \eqref{master-l-def} with regards to the
sum of the endogenous assets owned by the underwriter itself, the assumption is still liberal
in the sense that no restrictions
on derivatives regarding exogenous assets have been made (see also Eq.~\eqref{liabilities a} and 
the remarks there). In that sense, we could as well write $\psi^i_j={_{\mathbf{a}}\psi^i_j}$ 
in Assumption \ref{master-def_2}. For example, liabilities, where 
$\mathbf{d}^i_{\mathbf{r}^1,\ldots,\mathbf{r}^m,\mathbf{s},\mathbf{a}} \equiv 
\mathbf{b}^i(\mathbf{a})\in (\R^+_0)^n$, fulfill Assumption \ref{master-def_2}
(see also Sec.~\ref{Example with stocks, bonds, and derivatives} for an example). 

The following theorem is this paper's main result on the existence and uniqueness of no-arbitrage
prices in the presence of cross-ownership of equities and general liabilities of differing seniority.

\begin{theorem}
\label{master-theo}
Under Assumption \ref{master-def_1}, the following hold:
\begin{enumerate}
\item The system \eqref{master_n} -- \eqref{master_1} can only have non-negative solutions. 
\item For any solution of \eqref{master_n} -- \eqref{master_1}, the size of the sum of all balance
sheets is less than or equal to $(L^{\max}+1)||\mathbf{a}||_1$, where $L^{\max}$ is as in \eqref{L^max}.
\item Under the additional Assumption \ref{assumption_1}, the system 
\eqref{master_n} -- \eqref{master_1} has at least one solution. 
\item Under the additional Assumption \ref{master-def_2}, the solution of
\eqref{master_n} -- \eqref{master_1} is unique, i.e.~all 
endogenous assets are derivatives of the exogenous assets. The implicit function
\begin{equation}
\label{derivative}
\Psi:
\mathbf{a}
\longrightarrow 
\left(
\begin{array}{*{3}{l}}
\mathbf{r}^{1*}(\mathbf{a})\\
\vdots \\
\mathbf{r}^{m*}(\mathbf{a})\\
\mathbf{s}^*(\mathbf{a})
\end{array}
\right) 
\end{equation}
(the `derivative') that maps the exogenous assets $\mathbf{a}$ on the solution of 
\eqref{master_n} -- \eqref{master_1} is Lebesgue-measurable.
\end{enumerate}
\end{theorem}

The first part of Theorem \ref{master-theo} is Lemma \ref{master-pos_sols}. 
The second part follows directly from Eq.~\eqref{upper_balance_boundary}.
Proofs for the other parts and an algorithm for the solution in part four can be found in 
Section \ref{Proof and algorithm}.
The fourth part means that under Assumption \ref{master-def_2}, 
similarly as in the original Merton model, all (recovery) claims on equities and
liabilities are derivatives of the underlying exogenously priced assets. 
The payout of a derivative is, as always, understood to be a
function of the underlying asset(s), i.e.~the payout is {\em uniquely} determined by the value of the
underlying. The existence of a unique solution for the system \eqref{master_n} -- \eqref{master_1}
therefore implies the function \eqref{derivative}.
For instance, if $\mathbf{a}$ was replaced by a vector of random variables (i.e.~a stochastic model
for the endogenously priced assets), then this derivative
would be a random variable too due to Lebesgue-measurability.
Obviously, it is here assumed that the ownership matrices $\mathbf{M}^s$ and $\mathbf{M}^{d,1},\ldots,\mathbf{M}^{d,m}$ and other features of \eqref{master_n} -- \eqref{master_1}
are fixed -- as they would be in a real life situation. 
Since $\Psi$ is a Lebesgue-measurable 
derivative of $\mathbf{a}$ that uniquely determines the no-arbitrage prices of equity and recovery
claims in terms of the exogenous assets, this means that, as in the original Merton valuation model
for one firm, the theory of risk-neutral pricing can be applied for the valuation of equities and
liabilities in the multi-firm case with cross-ownership. In particular, this means that no-arbitrage prices 
{\em before} maturity can be determined (see Sec.~\ref{Equilibria for non-zero maturities}).

Clearly, Assumption \ref{master-def_2} is a restriction that one would want to weaken as much as possible. 
However, for uniqueness, Assumption \ref{assumption_1} alone is not
enough, as an example with more than one no-arbitrage equilibrium in Section \ref{No unique equilibrium} will show
(see also Sec.~\ref{A comment on uniqueness}).


\subsection{A comment on liquidation and no-arbitrage}
\label{A comment on solvency, liquidation and arbitrage}

Equation \eqref{balance sheet} implies that
\begin{equation}
\label{netting}
\mathbf{a} + \sum_{i=1}^{m} \mathbf{M}^{d,i}\cdot\mathbf{r}^i - \sum_{i=1}^{m} \mathbf{r}^i
= 
\mathbf{s} - \mathbf{M}^s\cdot\mathbf{s} .
\end{equation}
Netting or paying all recovery claims (payable liabilities) at maturity is
no problem if $\mathbf{M}^s=\mathbf{0}$, even if some (or all) firms are in default.
The reason is that in this case the right side of \eqref{netting} is non-negative, which
means that the left side implies that, for any of the $n$ companies, exogenous assets plus receivable 
recovery claims are sufficient to pay the recovery claims on all liabilities they have written.
Hence, no sale or liquidation of equity is needed to pay for the cash liabilities. 

If for any company the corresponding row of \eqref{netting} was identical zero,
then exogenous assets and receivables would exactly cover the payable liabilities -- a sign that
the firm's equity was zero and that it might be in default.

Problems arise if some components of the right side of \eqref{netting} are negative,
since in this case selling the exogenous assets (which can be treated like cash
for this purpose) and receiving cash claims will not be enough to pay the recovery claims belonging
to the liabilities of these firms. This means that assets in the form of equity, 
i.e.~components of $\mathbf{M}^s\cdot\mathbf{s}$, have to be liquidated
in order to pay for liabilities. 
This is only possible if there either is a buyer (inside or outside the group of firms) willing to buy
at the determined no-arbitrage prices, or if equity owners (shareholders) are able to
cash in on equity by liquidating their share of the firm's assets.
In Section \ref{The equilibrium equations} this problem was avoided by considering directly
liquidation values, i.e.~implicitly it was assumed that all $n$ firms got liquidated at maturity.
Such a liquidation could happen as follows. A bank, e.g.~a central bank, 
lends each firm an amount of the size of the firm's balance sheet according to
no-arbitrage prices. Each company then pays out liabilities and
equity as summarized by the right side of the balance sheet \eqref{balance sheet}. Hence, each company
receives equities and claims as listed on the left side of the balance sheet, and it can also sell
the exogenous assets. Each bank loan is then paid back using the proceeds (left side
equals right side). It is assumed that this happens instantaneously and no interest is paid
on the bank loans. The result in the end is that all values in liabilites and equities as determined
by the liquidation value equations \eqref{master_n} -- \eqref{master_1}
have been paid off. Note that in this case of a bank-buffered 
liquidation there is no need to sell any assets other than the exogenous ones outside the 
group of considered firms.

We now turn to an alternative approach of justifying the liquidation value equations 
\eqref{master_n} -- \eqref{master_1} as equations for prices under no-arbitrage by a 
no-arbitrage argument in which a complete liquidation is
not necessarily required. However, at least the {\em threat} of (partial) liquidation is
necessary to make the argument work.

\begin{assumption}
\label{equilibrium assumption 2}
As long as there are outstanding liabilities of a firm $i$ $(i=1,\ldots,m)$, any remaining assets
of firm $i$ have to be liquidated in order to pay these.
\end{assumption}

\begin{assumption}
\label{equilibrium assumption 3}
Because $\mathbf{M}^s$ and $\mathbf{M}^{d,i}$ $(i=1,\ldots,m)$ are {\em strictly} left substochastic 
matrices, we assume that there is a market {\em outside} the group of $n$ firms that instantaneously
trades fractions of the liabilities and equities that are not owned within the group of $n$ firms
at maturity. In this market, {\em any} amount of liabilities or equities of the group of $n$ firms,
or perfect replications thereof, can be sold at prevailing
prices. Assume that this market also trades the exogenous assets and that shortselling is allowed. 
Any amount of cash can be borrowed instantaneously at no interest.
Further, assume that this market is free of (instantaneous) arbitrage.
\end{assumption}

\begin{assumption}
\label{equilibrium assumption 4}
Any participant in the market of Assumption \ref{equilibrium assumption 3} who owns equity of any
of the $n$ firms is allowed to liquidate her part of the equity. 
\end{assumption}

\begin{assumption}
\label{equilibrium assumption 5}
The group of $n$ firms participates in the market of Assumption \ref{equilibrium assumption 3}
and sells any amount of equity at prevailing prices if asked to do so. Furthermore, any participant
in the market who owns the entire equity of any of the $n$ firms can liquidate this firm. 
\end{assumption}

We now check that the Assumptions \ref{equilibrium assumption 1},
\ref{equilibrium assumption 2}, \ref{equilibrium assumption 3}, and \ref{equilibrium assumption 4}
(first set of assumptions), alternatively the Assumptions  \ref{equilibrium assumption 1},
\ref{equilibrium assumption 2}, \ref{equilibrium assumption 3}, and 
\ref{equilibrium assumption 5} (second set of assumptions), 
lead to the equations \eqref{master_n} -- \eqref{master_1}
for the no-arbitrage prices of equities and liabilities in the market of 
Assumption \ref{equilibrium assumption 3}.

Under both sets of assumptions, the right hand sides of the no-arbitrage equations
\eqref{master_n} -- \eqref{master_1} correctly describe the payoffs of any liabilities, 
respectively the liquidation value of any equity. Therefore, no-arbitrage prices of any recovery
claims, which are paid in cash, have to equal the amounts paid. Hence, 
equations \eqref{master_n} and \eqref{master_j} apply to these prices.
For equities, the case where equity costs less than the liquidation value directly leads
to arbitrage by liquidation due to Assumption \ref{equilibrium assumption 4} or
Assumption \ref{equilibrium assumption 5}. However, in the case of 
Assumption \ref{equilibrium assumption 5}, the market agent would have to borrow enough cash
to buy {\em all} the equity of the corresponding firm, while under Assumption \ref{equilibrium assumption 4}
any fractional amount would suffice. If equity is overvalued compared to the 
liquidation value (right hand side of Eq.~\eqref{master_1} for the row corresponding to that firm), 
then a participant can
(proportionally) replicate the right hand side of Eq.~\eqref{master_1} and sell this
product in the market at higher than purchasing price, creating instant arbitrage. 
For this argument we have to assume that market participants are rational in the sense that
they do not discriminate between actual equity and perfectly replicated equity
(cf.~Assumption \ref{equilibrium assumption 3}).

In summary, it seems appropriate to say that the equations 
\eqref{master_n} -- \eqref{master_1} describe the prices of equities and liabilities
under the assumption of {\em no instantaneous arbitrage} given the possibility of (partial) liquidation.
The no-arbitrage arguments that were used obviously relate to Merton's comment on capital-structure
arbitrage in Mitchell (2004) (cf.~Section \ref{Cross-ownership: a two-firm example}).

It is clear that a violation of either the no-arbitrage equations \eqref{master_n} -- \eqref{master_1}, 
or the balance sheet equations \eqref{balance sheet} would mean either the violation of contractual
law (the priority of claims), or the manipulation
of balance sheets against better knowledge, for instance by ignoring current market prices of the
assets on that balance sheet, creating an opportunity of instantaneous arbitrage. In this sense, 
a no-arbitrage equilibrium determined by \eqref{master_n} -- \eqref{master_1} is indisputable from a 
theoretical point of view. However, a main ingredient of our argument is the possibility of
a liquidation, potentially a liquidation of all firms involved. This is of course also true for the
original Merton approach, but seems less rigorous there since only one firm with a very simple
capital structure is affected.
Nonetheless, we would argue that any pricing approach conflicting with the one outlined in this paper
could lead to significant problems at some stage since inconsistent pricing at one instance could 
lead to even bigger inconsistencies over longer periods of time.


\subsection{A comment on uniqueness}
\label{A comment on uniqueness}

We will see in Section \ref{No unique equilibrium} that there are examples when
Assumption \ref{assumption_1} holds but no unique price equilibrium exists.
In such cases, endogenous asset values are no derivatives of the underlying exogenous assets 
in the sense of an implicit function $\Psi$ as in \eqref{derivative}, since the considered relation is not
uniquely determined on the right hand side.
In such situations, the market can not price the assets in an indisputable and rational manner. 
Prices at maturity but also before maturity are indeterminable. 
This is a very unhealthy market situation because of the uncertainty
it creates. It is our opinion that financial contracts or derivatives that cause such situations
should be illegal.
In this sense, this theory could be useful to regulators in order to outrule such or similar 
situations and in order to assess these kind of situations before they become a real world problem.

Note that this problem is different from the problem of the incompleteness of a market.
In the case of incompleteness, payoffs, that can not be replicated, can have a range of prices
within which no arbitrage can occur, i.e.~the no-arbitrage price is not unique, and usually the
choice of one equivalent martingale measure (EMM) out of a range of possible EMMs leads to one
specific no-arbitrage price system. In the case of the example above,
even the choice of one specific EMM is of no use since there exists no uniquely determined
expectation of the discounted payoffs, because the payoffs themselves are not uniquely determined given
any state of the world. 
However, the situation is comparable with incompleteness to some extent since, 
in the case of non-unique no-arbitrage prices, the counterparties could agree on one specific 
no-arbitrage price set for each possible scenario of the exogenous assets, 
similar to the choice of one equivalent martingale measure in the incomplete market case. 
It is clear that the task of determining the set of all
priceable or allowed liabilities (derivatives) in a given system is equivalent to the question of the
weakest conditions for the liabilities functions
$\mathbf{d}^i$ that still allow for a unique solution of the no-arbitrage equations 
\eqref{master_n} -- \eqref{master_1}.


\section{Proof and algorithm}
\label{Proof and algorithm}


\subsection{Third part of Theorem \ref{master-theo}}

Important ingredients for the proof of the third part of Theorem \ref{master-theo} are the
Brouwer--Schauder Fixed Point Theorem (Theorem \ref{Brouwer} in Section 
\ref{Banach's Contraction Theorem}) and Lemma \ref{remark1} 
of Section \ref{A result for the sum of certain differences}.

\begin{proof}
Any solution of the system \eqref{master_n} -- \eqref{master_1} is a fixed point of 
the mapping $\Phi:$ 
\begin{equation}
\label{Phi}
\left(
\begin{array}{*{3}{l}}
\mathbf{r}^1\\
\mathbf{r}^2\\
\vdots \\
\mathbf{r}^m\\
\mathbf{s}
\end{array}
\right) 
\longmapsto 
\left(\!\!
\begin{array}{*{5}{l}}

\min\{\mathbf{d}^1_{\mathbf{r}^1,\ldots,\mathbf{r}^m,\mathbf{s}},
\;\; \mathbf{a} + \mathbf{M}^s\cdot\mathbf{s}  
+ \sum_{i=1}^{m} \mathbf{M}^{d,i}\cdot\mathbf{r}^i \}\\

\min\{\mathbf{d}^{2}_{\mathbf{r}^1,\ldots,\mathbf{r}^m,\mathbf{s}},\;
(\mathbf{a} + \mathbf{M}^s\cdot\mathbf{s}  
+ \sum_{i=1}^{m} \mathbf{M}^{d,i}\cdot\mathbf{r}^i 
- \mathbf{d}^1_{\mathbf{r}^1,\ldots,\mathbf{r}^m,\mathbf{s}})^+ \} \\

\vdots \\

\min\{\mathbf{d}^{m}_{\mathbf{r}^1,\ldots,\mathbf{r}^m,\mathbf{s}},\;
(\mathbf{a} + \mathbf{M}^s\cdot\mathbf{s}  
+ \sum_{i=1}^{m} \mathbf{M}^{d,i}\cdot\mathbf{r}^i 
- \sum_{i=1}^{m-1} \mathbf{d}^i_{\mathbf{r}^1,\ldots,\mathbf{r}^m,\mathbf{s}} )^+ \} \\

\qquad \qquad\qquad\;\;\, (\mathbf{a} + \mathbf{M}^s\cdot\mathbf{s}  
+ \sum_{i=1}^{m} \mathbf{M}^{d,i}\cdot\mathbf{r}^i 
- \sum_{i=1}^{m} \mathbf{d}^i_{\mathbf{r}^1,\ldots,\mathbf{r}^m,\mathbf{s}} )^+ 
\end{array}
\!\!\right) 
\end{equation}
and vice versa.
Because of Lemma \ref{master-pos_sols}, we only need to consider $\Phi$ on $(\mathbb{R}_0^+)^{n(m+1)}$.
Because of Eq.~\eqref{capcons} in Lemma \ref{remark1} (cf.~Section 
\ref{The balance sheet equation}), one has with $I^{\max}$ as in \eqref{I^max} and for
$\mathbf{r}^1,\mathbf{r}^2,\ldots,\mathbf{r}^m,\mathbf{s}\in (\mathbb{R}_0^+)^{n}$ that
\begin{eqnarray}
\label{Phi_norm}
\left|\left| 
\Phi
\left(
\begin{array}{l}
\mathbf{r}^{1}\\
\vdots \\
\mathbf{r}^{m}\\
\mathbf{s}
\end{array}
\right) 
\right|\right|_1 
& = & 
\left|\left| \mathbf{a} + \mathbf{M}^s\cdot\mathbf{s}  
+ \sum_{i=1}^{m} \mathbf{M}^{d,i}\cdot\mathbf{r}^i \right|\right|_1 \\
\nonumber & = &
||\mathbf{a}||_1 +  \left|\left|
\mathbf{M}^s\cdot \mathbf{s} \right|\right|_1
+ \sum_{i=1}^{m} \left|\left|\mathbf{M}^{d,i}\cdot \mathbf{r}^{i}
\right|\right|_1\\
\nonumber & \stackrel{\eqref{l_1}}{\leq} &
||\mathbf{a}||_1 + \left|\left|
\mathbf{M}^s\right|\right|_1\cdot \left|\left|\mathbf{s} \right|\right|_1
+ \sum_{i=1}^{m} \left|\left|\mathbf{M}^{d,i}\right|\right|_1\cdot
\left|\left|\mathbf{r}^{i}
\right|\right|_1\\
\nonumber & \leq & ||\mathbf{a}||_1 + I^{\max}\cdot 
\left(\left|\left|\mathbf{s} \right|\right|_1
+ \sum_{i=1}^{m}\left|\left|\mathbf{r}^{i}\right|\right|_1\right)\\
\nonumber & = & 
||\mathbf{a}||_1 +  
I^{\max}\cdot 
\left|\left|
\left(
\begin{array}{l}
\mathbf{r}^{1}\\
\vdots \\
\mathbf{r}^{m}\\
\mathbf{s}
\end{array}
\right) 
\right|\right|_1  .
\end{eqnarray}
From \eqref{Phi_norm}, and since $L^{\max} + 1 = \frac{1}{1+I^{\max}}$, we furthermore obtain that if
\begin{equation}
\label{Brouwer_1}
\mathbf{0}
\;\leq\;
\left|\left|
\left(
\begin{array}{l}
\mathbf{r}^{1}\\
\vdots \\
\mathbf{r}^{m}\\
\mathbf{s}
\end{array}
\right) 
\right|\right|_1 
\;\leq \;
(L^{\max}+1)||\mathbf{a}||_1  
\end{equation}
then
\begin{equation}
\label{Brouwer_2}
\mathbf{0}
\;\leq \;
\left|\left| 
\Phi
\left(
\begin{array}{l}
\mathbf{r}^{1}\\
\vdots \\
\mathbf{r}^{m}\\
\mathbf{s}
\end{array}
\right) 
\right|\right|_1 
\;\leq \;
(L^{\max}+1)||\mathbf{a}||_1 .
\end{equation}
Under Assumption \ref{assumption_1}, $\Phi$ (cf.~\eqref{Phi}) is continuous and \eqref{Brouwer_1} and 
\eqref{Brouwer_2} mean that we can apply the Brouwer--Schauder Fixed Point Theorem by considering $\Phi$
on the compact subset of $(\R_0^+)^{n(m+1)}$ defined by \eqref{Brouwer_1}. The system 
\eqref{master_n} -- \eqref{master_1} has 
therefore at least one non-negative solution with a balance sheet size of less than or equal to
$(L^{\max}+1)||\mathbf{a}||_1$.
\end{proof}


\subsection{Fourth part of Theorem \ref{master-theo} and algorithm}
\label{Fourth part of Theorem and algorithm}

The main ingredients for the proof of the fourth part of Theorem \ref{master-theo} are 
Banach's Contraction Mapping
Theorem, Lemma \ref{master-contractions_lemma} below, and Lemma \ref{stuff-nn} 
of Section \ref{A result for the sum of certain differences}.
Readers who are less familiar with the Contraction Mapping Theorem and the notion of a
contraction in general
might want to have a look at Section \ref{Banach's Contraction Theorem} 
(Definition \ref{Contraction} and Theorem \ref{Banach}) before reading on.

\begin{lemma}
\label{master-contractions_lemma}
Under Assumptions \ref{master-def_1} and \ref{master-def_2}, the mapping $\Phi$ 
in \eqref{Phi} is a strict contraction on $\R^{n(m+1)}$ under the metric implied by the 
$\ell^1$-norm on $\R^{n(m+1)}$.
\end{lemma}

\begin{proof}
Let $\mathbf{r}^{1,1},\ldots,\mathbf{r}^{m,1},\mathbf{s}^1,
\mathbf{r}^{1,2},\ldots,\mathbf{r}^{m,2},\mathbf{s}^2 \in \R^n$. 
Define $\mathbf{g}\in \R^{n(m+1)}$ as
\begin{equation}
\mathbf{g} 
=
\Phi
\left(
\begin{array}{l}
\mathbf{r}^{1,1}\\
\vdots \\
\mathbf{r}^{m,1}\\
\mathbf{s}^1
\end{array}
\right) 
\; - \;
\Phi
\left(
\begin{array}{l}
\mathbf{r}^{1,2}\\
\vdots \\
\mathbf{r}^{m,2}\\
\mathbf{s}^2
\end{array}
\right) ,
\end{equation}
and $\mathbf{h}\in \R^n$ as
\begin{eqnarray}
\mathbf{h} & = & \mathbf{M}^s\cdot\mathbf{s}^1  + \sum_{i=1}^{m} \mathbf{M}^{d,i}\cdot\mathbf{r}^{i,1} 
- \left(\mathbf{M}^s\cdot\mathbf{s}^2  + \sum_{i=1}^{m} \mathbf{M}^{d,i}\cdot\mathbf{r}^{i,2}\right) \\
\nonumber & = & 
\mathbf{M}^s\cdot (\mathbf{s}^1-\mathbf{s}^2)  
+ \sum_{i=1}^{m} \mathbf{M}^{d,i}\cdot (\mathbf{r}^{i,1}-\mathbf{r}^{i,2}) .
\end{eqnarray}
For $k\in\{1\ldots,n\}$, define $x=a_k$, and for $l=1,2$ define
\begin{equation}
y^l = \sum_{j=1}^{n} M^s_{kj}s^l_j  + \sum_{i=1}^{m} \sum_{j=1}^{n} M^{d,i}_{kj}r^{i,l}_j . 
\end{equation}
This means that $h_k = y^1 - y^2$ and that
\begin{eqnarray}
g_{k} & = & \min\left\{\psi^{1}_k(y^1),\; x + y^1 \right\} - \min\left\{\psi^{1}_k(y^2),\; x + y^2 \right\}\\
\nonumber & & \mbox{For $0 < j < m$:} \\
g_{k+nj} & = & \min\left\{\psi^{j+1}_k(y^1),\; \left(x + y^1 - \sum_{i=1}^{j} \psi^i_k(y^1) \right)^+ \right\} \\
\nonumber & & - \min\left\{\psi^{j+1}_k(y^2),\; \left(x + y^2 - \sum_{i=1}^{j} \psi^i_k(y^2) \right)^+ \right\}\\ 
g_{k+nm} & = & \left(x + y^1 - \sum_{i=1}^{m} \psi^i_k(y^1) \right)^+ 
- \left(x + y^2 - \sum_{i=1}^{m} \psi^i_k(y^2) \right)^+
\end{eqnarray}
Lemma \ref{stuff-nn} in Section \ref{A result for the sum of certain differences}, in conjunction with
Assumption \ref{master-def_2}, now implies that
\begin{equation}
\sum_{j=0}^{m} |g_{k+nj}| = |h_k| .
\end{equation}
Therefore $||\mathbf{g}||_1 = ||\mathbf{h}||_1$. With $I^{\max}$ as in \eqref{I^max} and \eqref{I<=I^max},
\begin{eqnarray}
||\mathbf{g}||_1 
& = & \left|\left| 
\Phi
\left(
\begin{array}{l}
\mathbf{r}^{1,1}\\
\vdots \\
\mathbf{r}^{m,1}\\
\mathbf{s}^1
\end{array}
\right) 
\; - \;
\Phi
\left(
\begin{array}{l}
\mathbf{r}^{1,2}\\
\vdots \\
\mathbf{r}^{m,2}\\
\mathbf{s}^2
\end{array}
\right)
\right|\right|_1 
\; = \; ||\mathbf{h}||_1 \\
\nonumber & = & 
\left|\left|
\mathbf{M}^s\cdot (\mathbf{s}^1-\mathbf{s}^2)  
+ \sum_{i=1}^{m} \mathbf{M}^{d,i}\cdot (\mathbf{r}^{i,1}-\mathbf{r}^{i,2})
\right|\right|_1\\
\nonumber & \leq &
\left|\left|
\mathbf{M}^s\cdot (\mathbf{s}^1-\mathbf{s}^2) \right|\right|_1
+ \sum_{i=1}^{m} \left|\left|\mathbf{M}^{d,i}\cdot (\mathbf{r}^{i,1}-\mathbf{r}^{i,2})
\right|\right|_1\\
\nonumber & \leq &
\left|\left|
\mathbf{M}^s\right|\right|_1\cdot \left|\left|(\mathbf{s}^1-\mathbf{s}^2) \right|\right|_1
+ \sum_{i=1}^{m} \left|\left|\mathbf{M}^{d,i}\right|\right|_1\cdot
\left|\left|(\mathbf{r}^{i,1}-\mathbf{r}^{i,2})
\right|\right|_1\\
\nonumber & \leq & I^{\max}\cdot 
\left(\left|\left|(\mathbf{s}^1-\mathbf{s}^2) \right|\right|_1
+ \sum_{i=1}^{m}\left|\left|(\mathbf{r}^{i,1}-\mathbf{r}^{i,2})\right|\right|_1\right)\\
\nonumber & = & 
I^{\max}\cdot 
\left|\left|
\left(
\begin{array}{l}
\mathbf{r}^{1,1}\\
\vdots \\
\mathbf{r}^{m,1}\\
\mathbf{s}^1
\end{array}
\right) 
\; - \;
\left(
\begin{array}{l}
\mathbf{r}^{1,2}\\
\vdots \\
\mathbf{r}^{m,2}\\
\mathbf{s}^2
\end{array}
\right)
\right|\right|_1  .
\end{eqnarray}
\end{proof}

\subsubsection*{Proof and algorithm}

\begin{proof}
Given Lemma \ref{master-contractions_lemma}, Banach's Contraction Mapping Theorem 
(Section \ref{Banach's Contraction Theorem}, Theorem \ref{Banach}) immediately implies 
that the system \eqref{master_n} -- \eqref{master_1} has a unique solution.
To obtain the unique solution of system \eqref{master_n} -- \eqref{master_1} under Assumption
\label{Algorithm for the unique solution in Theorem}
\ref{master-def_2}, the recursion \eqref{algo} of the Contraction Mapping Theorem 
(also sometimes called {\em Picard iteration}) can be used with $\Phi$ as in \eqref{Phi}. 
To better account for the dependency of \eqref{Phi} on $\mathbf{a}$,
we denote \eqref{Phi} now by $\Phi_{\mathbf{a}}$. Hence,
\begin{equation}
\label{master-algo}
\Psi(\mathbf{a}) \equiv 
\lim_{m\rightarrow \infty} \Phi^m_{\mathbf{a}}(\cdot) = \lim_{m\rightarrow \infty} 
\underbrace{\Phi_{\mathbf{a}}\circ\ldots\circ \Phi_{\mathbf{a}}}_m (\cdot) . 
\end{equation}
Regarding the proof of Lebesgue-measurability, it is straightforward to show that the function
$\Phi_{\cdot}(\cdot)$ is a continuous map $\R^n\times\R^{n(m+1)} \rightarrow \R^{n(m+1)}$.
Similarly, it is straigtforward that the function (note the somewhat sloppy notation)
$\Phi_{\cdot}^m(\cdot): \R^n\times\R^{n(m+1)} \rightarrow \R^{n(m+1)}$ ($m=1,2,\ldots$), which is 
obtained from 
\begin{equation}
\label{kd}
\Phi_{\mathbf{a}}^m(\cdot) = \underbrace{\Phi_{\mathbf{a}}\circ\cdots\circ\Phi_{\mathbf{a}}}_{m}(\cdot) 
\end{equation}
by replacing $\mathbf{a}$ by $\cdot$, 
is continuous too. Continuity implies Lebesgue-measurability. However, 
we know that $\Phi_{\mathbf{a}}^m(\cdot)$ has the non-functional limit
\eqref{master-algo}, hence $\lim_{m\rightarrow \infty} \Phi_{\cdot}^m(\cdot) = \Psi_{\cdot}$.
A point-wise limit of measurable functions is measurable, 
therefore $\Psi_{\cdot}=\Psi_{\cdot}(\cdot)$
is Lebesgue-measurable on $\R^n\times\R^{n(m+1)}$, with the second argument being irrelevant. 
It follows now from basic measure theory that $\Psi_{\cdot}$ is Lebesgue-measurable on $\R^n$.
\end{proof}


\section{Applications and examples}
\label{Examples}


\subsection{No-arbitrage prices before maturity}
\label{Equilibria for non-zero maturities}

So far, only the existence of no-arbitrage prices {\em at} maturity has been considered. Assume
now that the exogenous assets are given by an $n$-dimensional price process $\mathbf{a}(t)$ where 
$t\in\mathbb{T}$, with $\min \mathbb{T} = 0$ being the present time,
and where $\max \mathbb{T} = T$ is the time of maturity. 
Irrespective of the particular structure of $\mathbb{T}$, we assume that
$\mathbf{a}(\cdot)$ is adapted to a filtration $(\mathcal{F}_t)_{t\in\mathbb{T}}$ that lives on a 
probability space $(\Omega, \F_T, \mathbb{P})$. For convenience, assume that $a_1(\cdot)$ is a 
num\'eraire, i.e.~non-dividend paying and almost surely strictly positive. Assume now that $\mathbb{Q}$
is an equivalent martingale measure (EMM) for the price process $\mathbf{a}(\cdot)$, i.e.~by the 
Fundamental Theorem of Asset Pricing, we assume
that the market spanned by $\mathbf{a}(\cdot)$ is arbitrage-free. Furthermore,
consider a payoff $X$ at maturity $T$, given by a random variable on $(\Omega, \F_T, \mathbb{P})$.
Assuming sufficient integrability, or even boundedness where necessary, the Fundamental Theorem
of Asset Pricing implies that for $t\in\mathbb{T}$ the risk-neutral valuation formula
\begin{equation}
\label{Risk-neutral Valuation Formula_1}
X(t) = a_1(t)\mathbf{E}_{\mathbb{Q}}\left[\left.\frac{X}{a_1(T)}\right|\mathcal{F}_t\right]
\end{equation}
defines an arbitrage-free price process $X(\cdot)$ for the payoff $X=X(T)$ at $T$. 
Because of Lebesgue-measurability, the implicit function $\Psi$
in \eqref{derivative} is a vector of payoffs like $X$ when evaluated at maturity using 
$\mathbf{a} = \mathbf{a}(T)$. Therefore, if uniqueness of the solution of 
\eqref{master_n} -- \eqref{master_1} is given for all possible outcomes of $\mathbf{a}(T)$,
e.g.~by part four of Theorem \ref{master-theo}, then the risk-neutral valuation formula 
\begin{equation}
\label{Risk-neutral Valuation Formula}
a_1(t)\mathbf{E}_{\mathbb{Q}}\left[\left.\frac{\Psi(\mathbf{a}(T))}{a_1(T)}\right|\mathcal{F}_t\right]
\end{equation}
provides
no-arbitrage prices of the equities and the liabilities of the considered system at any time $t\in\mathbb{T}$.
While completeness of the market given by the underlying exogenous assets' price processes plays no role 
in this consideration, uniqueness of the no-arbitrage prices at maturity (given any scenario for the 
underlying assets) is crucial. However, uniqueness of equities' and liabilities' no-arbitrage prices at
any time {\em before} maturity naturally depends on the replicability of $\Psi(\mathbf{a}(T))$ 
at maturity in terms of the exogenous assets, which would be given in a generally complete market.

Regarding the balance sheet considerations of Sec.~\ref{The balance sheet equation},
it is clear that the accounting equation \eqref{balance sheet} also holds when
$\mathbf{a}$, $\mathbf{s}$, and the $\mathbf{d}^i$ are replaced by corresponding no-arbitrage
prices before maturity. This follows directly from the application of 
\eqref{Risk-neutral Valuation Formula_1} to both sides of \eqref{balance sheet}.

Regarding the measures of leverage and cross-ownership in Sec.~\ref{Measures of leverage and cross-ownership},
similarly to the balance sheet equation case, in \eqref{L} -- \eqref{MMLast} one could substitute
$\mathbf{a}$, $\mathbf{s}$, and the $\mathbf{d}^i$ by corresponding
no-arbitrage prices before maturity. In general, this would lead to results different from 
\eqref{measure before maturity}.


\subsection{More exogenous assets than firms}
\label{Extended exogenous assets}

So far we have assumed the exogenous assets to be given by a vector $\mathbf{a}\in (\mathbb{R}_0^+)^n$,
where the dimension $n$ is the number of firms considered in the system. This rather artificial assumption,
which was entirely sufficient for our considerations at maturity conditional on one price scenario
of the exogenous assets, can naturally be extended in the
following way. We assume that the exogenous assets are given by a vector 
$\mathbf{a}\in (\mathbb{R}_0^+)^q$ where $q\in\{1,2,\ldots\}$. To describe the ownership of these
assets we need an ownership matrix $\mathbf{M}^a\in\mathbb{R}^{n\times q}$ with column sums in $[0,1]$. 
Similar to the obvious change in the liquidation value
(no-arbitrage price) equations \eqref{master_n} -- \eqref{master_1}, 
which become
\begin{eqnarray}
\label{master_na}
\mathbf{r}^1 & = & \min\left\{\mathbf{d}^1_{\mathbf{r}^1,\ldots,\mathbf{r}^m,\mathbf{s},\mathbf{a}},
\quad\, \mathbf{M}^a\cdot\mathbf{a} + \mathbf{M}^s\cdot\mathbf{s}  
+ \sum_{i=1}^{m} \mathbf{M}^{d,i}\cdot\mathbf{r}^i \right\}\\
\nonumber & & \mbox{For $0 < j < m$:} \\
\label{master_ja} 
\mathbf{r}^{j+1} & = & \min\left\{\mathbf{d}^{j+1}_{\mathbf{r}^1,\ldots,\mathbf{r}^m,\mathbf{s},\mathbf{a}},\;
\left(\mathbf{M}^a\cdot\mathbf{a} + \mathbf{M}^s\cdot\mathbf{s}  
+ \sum_{i=1}^{m} \mathbf{M}^{d,i}\cdot\mathbf{r}^i 
- \sum_{i=1}^{j} \mathbf{d}^i_{\mathbf{r}^1,\ldots,\mathbf{r}^m,\mathbf{s},\mathbf{a}} \right)^+ \right\} \\
\label{master_1a}
\mathbf{s} & = & \qquad\qquad\qquad\quad\quad \left(\mathbf{M}^a\cdot\mathbf{a} + \mathbf{M}^s\cdot\mathbf{s}  
+ \sum_{i=1}^{m} \mathbf{M}^{d,i}\cdot\mathbf{r}^i 
- \sum_{i=1}^{m} \mathbf{d}^i_{\mathbf{r}^1,\ldots,\mathbf{r}^m,\mathbf{s},\mathbf{a}} \right)^+ ,
\end{eqnarray}
the theory as described so far stays exactly the same, or can be adapted in the entirely obvious way by
replacing $\mathbf{a}$ with $\mathbf{M}^a\cdot\mathbf{a}$, 
including Theorem \ref{master-theo} and Section \ref{Equilibria for non-zero maturities} on 
prices before maturity. However, using the extended set-up
allows for even more general liabilities (derivatives) due to 
$\mathbf{d}^i_{\mathbf{r}^1,\ldots,\mathbf{r}^m,\mathbf{s}} = 
\mathbf{d}^i_{\mathbf{r}^1,\ldots,\mathbf{r}^m,\mathbf{s},\mathbf{a}}$ (cf.~Eq.~\eqref{liabilities a}).
See Sec.~\ref{Example with stocks, bonds, and derivatives} for an example.


\subsection{An example with stocks, bonds, and derivatives}
\label{Example with stocks, bonds, and derivatives}

Consider a system of two firms, $n=2$, and three exogenous assets, $\mathbf{a}\in(\R_0^+)^3$ and 
$\mathbf{M}^a\in\R^{2\times 3}$ (cf.~Sec.~\ref{Extended exogenous assets}). Suppose now for the
liabilities ($m=3$)
\begin{equation}
\label{example6.3}
\mathbf{d}^1
= 
\left(
\begin{array}{l}
b_1\\
b_3\\
\end{array}
\right) , \quad
\mathbf{d}^2_{\mathbf{a}}
= 
\left(
\begin{array}{l}
b_2\\
c_2(0.5a_1+a_3 - k_2)^-\\
\end{array}
\right) , \quad
\mathbf{d}^3_{\mathbf{a}}
= 
\left(
\begin{array}{l}
c_1(a_2-k_1)^+\\
0\\
\end{array}
\right) ,
\end{equation}
where $b_1,b_2,b_3,c_1,c_2,k_1,k_2 > 0$.
As in Assumption \ref{master-def_1}, let $\mathbf{M}^s, \mathbf{M}^{d,i} \in\R^{2\times 2}$ $(i=1,2,3)$ be 
strictly left substochastic matrices. 
The described system is one of two firms where the first
one has issued two bonds of differing seniority (that of nominal $b_1$ higher than that of $b_2$),
and one derivative -- a European Call on exogenous asset $a_2$ with a strike price of $k_1$ and 
a `size' of $c_1$.
The second firm has issued one bond (nominal $b_3$) and one derivative -- a European Put
on a basket (a mix) of exogenous assets $a_1$ and $a_3$ with a strike price of $k_2$ and a `size'
of $c_2$. Additionally, and not specified in more detail, any level of cross-ownership of any
of these five liabilities and the two equities could be present.
It is clear in this case that part four of Theorem \ref{master-theo} applies, with the unique 
no-arbitrage price equilibrium 
given by Eq.~\eqref{master_na} -- \eqref{master_1a}. These no-arbitrage prices could be calculated using
the algorithm \eqref{master-algo} of Sec.~\ref{Fourth part of Theorem and algorithm}.
Using risk-neutral pricing techniques under a stochastic model for the exogenous assets
(cf.~Sec.~\ref{Equilibria for non-zero maturities}), one could therefore simultaneously calculate
no-arbitrage prices of all claims (equities, loans, derivatives) in this system, while fully accounting
for the priority of claims, as well as for leverage and counterparty risk caused by cross-ownership.

As a remark, in \eqref{example6.3} the second entries of $\mathbf{d}^2_{\mathbf{a}}$ and
$\mathbf{d}^3_{\mathbf{a}}$ could be swapped without any consequences for pricing. Furthermore, 
it is clear in what way the above example would simplify if one was only interested in the
valuation of equity, bonds and derivatives of differing seniority issued by one single firm,
free of any cross-ownership entanglements.


\subsection{No unique prices under Assumption \ref{assumption_1}}
\label{No unique equilibrium}

Suppose $n=2$ and $m=1$, $\mathbf{M}^s=\mathbf{0}$ and 
\begin{equation}
\mathbf{M}^{d} = \left(
\begin{array}{*{2}{c}}
\arraycolsep0pt
0 & 0.8\\
0.8 & 0 
\end{array}
\right) ,
\quad
\mathbf{d}_{\mathbf{r},\mathbf{s}} = 
\left(
\begin{array}{l}
(r_2-2)^2 \\
(r_1-2)^2
\end{array}
\right) ,
\quad
\mathbf{a} =  
\left(
\begin{array}{l}
1 \\
1
\end{array}
\right) .
\end{equation}
It can now easily be checked that the no-arbitrage equations
\begin{eqnarray}
r_1 & = & \min\{(r_2-2)^2 , 1 + 0.8r_2\} \\
r_2 & = & \min\{(r_1-2)^2 , 1 + 0.8r_1\} \\
s_1 & = & (1 + 0.8r_2 - (r_2-2)^2)^+ \\
s_2 & = & (1 + 0.8r_1 - (r_1-2)^2)^+ .
\end{eqnarray}
are solved by
\begin{equation}
\mathbf{r}^{*} =  
\left(
\begin{array}{l}
1 \\
1
\end{array}
\right) 
\quad\text{ and }\quad
\mathbf{s}^{*} =  
\left(
\begin{array}{l}
0.8 \\
0.8
\end{array}
\right) ,
\end{equation}
as well as by
\begin{equation}
\mathbf{r}^{*} =  
\left(
\begin{array}{l}
4 \\
4
\end{array}
\right) 
\quad\text{ and }\quad
\mathbf{s}^{*} =  
\left(
\begin{array}{l}
0.2 \\
0.2
\end{array}
\right) .
\end{equation}
Note that in this example \eqref{master-l-cond} does not hold.
For the theoretical meaning of this example see also Sec.~\ref{A comment on uniqueness}.


\subsection{No price equilibrium under maximum cross-ownership}
\label{No equilibrium under maximum cross-ownership}

Let $\mathbf{M}^s$ be a maximum ownership matrix in the sense that the matrix is left stochastic, 
that is each column adds up to $1$. Let $||\mathbf{a}||_1>0$, $\mathbf{d}^{i}\equiv\mathbf{0}$ and 
$\mathbf{M}^{d,i}=\mathbf{0}$ $(i=1,\ldots,m)$. 
Suppose now that a no-arbitrage equilibrium for this set-up exists. Because of the assumptions, 
$ \left|\left| \mathbf{s} \right|\right|_1 
= \left|\left| \mathbf{M}^s\cdot\mathbf{s} \right|\right|_1$.
Under no-arbitrage, \eqref{total_conservation_2} holds. Therefore,
\begin{equation}
\left|\left| \mathbf{s} \right|\right|_1 =
\left|\left| \mathbf{a} \right|\right|_1 + \left|\left| \mathbf{M}^s\cdot\mathbf{s} \right|\right|_1
= \left|\left| \mathbf{a} \right|\right|_1 + \left|\left| \mathbf{s} \right|\right|_1  ,
\end{equation}
which is a contradiction since $||\mathbf{a}||_1>0$. Also, since $I^{\max}=1$, one has
$L^{\max}=\frac{I^{\max}}{1-I^{\max}}=+\infty$.


\subsection{$L^{\max}$ is sharp}
\label{L_max is sharp}

Consider a system with $n=2$, $m=1$, $\mathbf{M}^{d}=\mathbf{0}$, $\mathbf{d}=\mathbf{0}$, 
\begin{equation}
\nonumber
\mathbf{M}^{s} = \left(
\begin{array}{*{2}{c}}
\arraycolsep0pt
0 & 0.5\\
0.5 & 0 
\end{array}
\right) ,
\quad
\text{and}
\quad
\mathbf{a} =  
\left(
\begin{array}{l}
1 \\
1
\end{array}
\right) .
\end{equation}
The no-arbitrage equations are
\begin{eqnarray}
\label{r2}
\mathbf{r} & = & \mathbf{0}\\
\label{s2}
\mathbf{s} & = & (\mathbf{a} + \mathbf{M}^s\cdot\mathbf{s})^+ .
\end{eqnarray}
We therefore have $L^{\max} = 1$, and the upper boundary of the sum of all balance sheets is 
$(L^{\max}+1)||\mathbf{a}||_1 = 4$. It is clear from Theorem \ref{master-theo} that we have
unique no-arbitrage prices in this set-up. It can easily be checked against \eqref{r2} and \eqref{s2}
that this price equilibrium is
given by $r_1 = r_2 =0$ and $s_1 = s_2 = 2$, which means that the sum of all balance sheets
equals $4$, which is the value of $(L^{\max}+1)||\mathbf{a}||_1$. Hence, $L^{\max}$ is sharp.


\section{Conclusion}

This paper has presented a model for the no-arbitrage valuation of equities and general liabilities in a
system of firms where cross-ownership is present and where liabilities, that can 
include debts and derivatives, can be of differing seniority in a liquidation. Cross-ownership is a 
widespread financial
phenomenon (cf.~Sec.~\ref{General liabilities under cross-ownership}), and as the presented
theory directly accounts for counterparty risk and, in a way, systemic risk, its 
valuation procedure should be relevant for the theory and practice of general asset valuation, 
derivatives pricing,
and credit risk management. For the application of the ideas of this paper, it might be helpful
that our theory is a direct extension of the Merton (1974) model (see also Sec.~\ref{Valued added}),
which is the basis of modern structural credit risk models. For instance, the theory presented in
this paper is general enough to be applied in stochastic interest rates settings, in settings where 
underlying exogenous assets are modelled with copula approaches, and for credit risk modelling, 
where specifically defined default barriers could be applied for the calculation of default probabilities.
Future directions of research to extend this theory should include
investigations on the range of liabilities that allow for unique no-arbitrage price equilibria 
(weaker forms of Assumption \ref{master-def_2}) and, of course, the question of no-arbitrage
pricing in the multi-period case.



\begin{appendix}


\section{Appendix}
\label{Technical results}


\subsection{A result for substochastic matrices}
\label{A result for substochastic matrices}

\begin{lemma}
\label{M-lemma}
If $\mathbf{M}\in \R^{n\times n}$ is a strictly (left or right) substochastic matrix, then
$(\mathbf{I}-\mathbf{M})^{-1}$ exists and is non-negative. The diagonal elements of 
$(\mathbf{I}-\mathbf{M})^{-1}$ are greater than or equal to 1.
\end{lemma}

\begin{proof}
For $||\mathbf{M}||_1 < 1$, it follows from standard results of functional analysis that
$(\mathbf{I}-\mathbf{M})^{-1}$ exists and
\begin{equation}
(\mathbf{I}-\mathbf{M})^{-1} = \sum_{n=0}^{\infty} \mathbf{M}^n .
\end{equation}
The lemma then follows from $\mathbf{M}^n \geq \mathbf{0}$ for $n=0,1,\ldots$, and 
$\mathbf{M}^0=\mathbf{I}$.
\end{proof}

\subsection{Two fixed point theorems}
\label{Banach's Contraction Theorem}

\begin{theorem}[Brouwer--Schauder Fixed Point Theorem]
\label{Brouwer}
Every continuous function from a convex compact subset $\mathbb{K}$ of a Banach space to 
$\mathbb{K}$ itself has a fixed point.
\end{theorem}

\begin{definition}
\label{Contraction}
Let $(\mathbb{X},d)$ be a metric space. A map $\Phi: \mathbb{X} \rightarrow \mathbb{X}$ 
is called a strict contraction on $\mathbb{X}$ if 
there exists a number $0 \leq c < 1$ such that
\begin{equation}
d(\Phi(\mathbf{x}),\Phi(\mathbf{y})) \leq c\cdot d(\mathbf{x}, \mathbf{y}) \quad \text{for} 
\quad \mathbf{x},\mathbf{y}\in \mathbb{X} .
\end{equation}
The map $\Phi$ is called a weak contraction if
\begin{equation}
d(\Phi(\mathbf{x}),\Phi(\mathbf{y})) \leq d(\mathbf{x}, \mathbf{y}) \quad \text{for} 
\quad \mathbf{x},\mathbf{y}\in \mathbb{X} .
\end{equation}
\end{definition}

\begin{theorem}[Banach Contraction Mapping Theorem]
\label{Banach}
Let $\mathbb{X}$ be a complete metric space and f be a strict contraction on $\mathbb{X}$. 
Then $\Phi$ has a unique fixed point $\mathbf{x}^* \in \mathbb{X}$. For any 
$\mathbf{x} \in \mathbb{X}$, one has
\begin{equation}
\label{algo}
\lim_{n\rightarrow \infty} \Phi^n(\mathbf{x}) = \lim_{n\rightarrow \infty} 
\underbrace{\Phi\circ\ldots\circ \Phi}_n (\mathbf{x}) = \mathbf{x}^* . 
\end{equation}
\end{theorem}


\subsection{Two lemmas}
\label{A result for the sum of certain differences}

\begin{lemma}
\label{remark1}
For $x\in\R$, $m\in\{1,2,\ldots\}$, and $y^1,\ldots,y^m\in\R_0^+$,
\begin{eqnarray}
\label{capcons}
x & = & \min\left\{y^1,x\right\}
\; + \; \sum_{j=1}^{m-1}\;\; \min\left\{y^{j+1},\left(x-\sum_{i=1}^{j}y^i\right)^+\right\} 
\; + \; \left(x-\sum_{i=1}^{m}y^i\right)^+ .
\end{eqnarray}
\end{lemma}
\begin{proof}
This is easy to check. 
\end{proof}

\begin{lemma}
\label{stuff-nn}
For $x,y^1,y^2\in\R$ and for monotonically increasing functions $\psi^i: \R \rightarrow \R^+_0$ $(i=1,\ldots,m)$
such that for any $z^1,z^2\in\R$ with $z^1 \geq z^2$
\begin{equation}
\label{l}
z^1-z^2 \; \geq \; \sum_{i=1}^{m} (\psi^i(z^1) - \psi^i(z^2))
\end{equation}
the following equation holds:
\begin{eqnarray}
\label{stuffnn}
|y^1-y^2| & = & 
\left|\min\left\{\psi^1(y^1),x+y^1\right\} - \min\left\{\psi^1(y^2),x+y^2\right\}\right|\\
\nonumber &  & + \; 
\sum_{j=1}^{m-1}\;\; \left|\min\left\{\psi^{j+1}(y^1),\left(x+y^1-\sum_{i=1}^{j}\psi^i(y^1)\right)^+\right\}\right. \\
\nonumber &  & \qquad\qquad 
- \; \left.\min\left\{\psi^{j+1}(y^2),\left(x+y^2-\sum_{i=1}^{j}\psi^i(y^2)\right)^+\right\}\right| \\
\nonumber &  & + \left|\left(x+y^1-\sum_{i=1}^{m}\psi^i(y^1)\right)^+ 
- \left(x+y^2-\sum_{i=1}^{m}\psi^i(y^2)\right)^+\right| .
\end{eqnarray}
\end{lemma}

Obviously, any of the functions $\psi^i$ in Lemma \ref{stuff-nn} could be a constant.

\begin{proof}
We will prove the equation considering six cases (with sub-cases) for which we will derive simplified
expressions for the right hand side of \eqref{stuffnn}. It will be fairly straightforward to check that
these expression are correct. Without loss of generality, assume $y^1 \geq y^2$, and hence
$\psi^i(y^1) \geq \psi^i(y^2)$ $(i=1,\ldots,m)$. Because of this,
{\em all} absolute expressions $|\cdot|$ below will be positive anyway. We keep the $|\cdot|$ for convenience. 
Regarding the aforementioned cases, note that for $y^1 \geq y^2$ and for
$j \leq k$, $j,k\in\{1,\ldots,m\}$, it is impossible to have a situation where simultaneously
\begin{eqnarray}
\label{c1} \sum_{i=1}^{j} \psi^{i}(y^1) & \geq & x + y^1 \\
\label{c2} \sum_{i=1}^{k} \psi^{i}(y^2) & \leq & x + y^2 .
\end{eqnarray}
This becomes immediately clear from a subtraction of the two inequalities, \eqref{c1}-\eqref{c2},
which leads to a contradiction of \eqref{l}:
\begin{equation}
y^1-y^2 \; \leq \; \sum_{i=1}^{j} \psi^{i}(y^1) - \sum_{i=1}^{k} \psi^{i}(y^2) 
\; \leq \; \sum_{i=1}^{k} (\psi^{i}(y^1) - \psi^{i}(y^2)) .
\end{equation}
We therefore only have to consider the following cases.\\
Case 1: $\sum_{i=1}^{m}\psi^i(y^1) \leq x+y^1$ and $\sum_{i=1}^{m}\psi^i(y^2) \leq x+y^2$ .
In this case, using \eqref{l}, the right side of \eqref{stuffnn} is identical to
\begin{equation}
\sum_{i=1}^{m}\; \left|\psi^{i}(y^1) - \psi^{i}(y^2)\right| 
\; + \;
\left|x+y^1-\sum_{i=1}^{m}\psi^i(y^1) - \left(x+y^2-\sum_{i=1}^{m}\psi^i(y^2)\right)\right|
\; = \; |y^1-y^2|
\end{equation}
Case 2: Assume that $m > k > 0$ and 
\begin{equation}
\sum_{i=1}^{m}\psi^i(y^1) \;\leq\; x+y^1 
\end{equation}
and
\begin{equation}
\sum_{i=1}^{k}\psi^i(y^2) \;\leq\; x+y^2 \;\leq\; \sum_{i=1}^{k+1}\psi^i(y^2) ,
\end{equation}
which implies
\begin{equation}
\label{nice}
0 \;\leq\; x+y^2 - \sum_{i=1}^{k}\psi^i(y^2) \;\leq\; \psi^{k+1}(y^2) \;\leq\; \psi^{k+1}(y^1) .
\end{equation}
\\
Case 2.1: $k = m-1$. The right hand side of \eqref{stuffnn} becomes
\begin{eqnarray} 
\sum_{i=1}^{m-1}\; \left|\psi^{i}(y^1) - \psi^{i}(y^2)\right| \; + \; 
\left|\psi^{m}(y^1) - \left(x+y^2-\sum_{i=1}^{m-1}\psi^i(y^2)\right)\right| & & \\
\nonumber \; + \; \left|x+y^1-\sum_{i=1}^{m}\psi^i(y^1)\right|  
& \stackrel{\eqref{nice}}{=} & |y^1-y^2| .
\end{eqnarray}
Case 2.2: $k \leq m-2$. The right hand side of \eqref{stuffnn} turns into
\begin{eqnarray}
\sum_{i=1}^{k}\; \left|\psi^{i}(y^1) - \psi^{i}(y^2)\right| 
\; + \;  \left|\psi^{k+1}(y^1) - \left(x+y^2-\sum_{i=1}^{k}\psi^i(y^2)\right)\right| & & \\
\nonumber 
\; + \;  |\psi^{k+2}(y^1)|   + \ldots  +  |\psi^{m}(y^1)|
\; + \; \left|x+y^1-\sum_{i=1}^{m}\psi^i(y^1)\right|  
& \stackrel{\eqref{nice}}{=} & |y^1-y^2| .
\end{eqnarray}
Case 3: Assume that $m > j \geq k > 0$ and 
\begin{equation}
\sum_{i=1}^{j}\psi^i(y^1) \;\leq\; x+y^1 \;\leq\; \sum_{i=1}^{j+1}\psi^i(y^1) 
\end{equation}
as well as
\begin{equation}
\sum_{i=1}^{k}\psi^i(y^2) \;\leq\; x+y^2 \;\leq\; \sum_{i=1}^{k+1}\psi^i(y^2) ,
\end{equation}
which again implies
\begin{equation}
\label{nice2}
0 \;\leq\; x+y^2 - \sum_{i=1}^{k}\psi^i(y^2) \;\leq\; \psi^{k+1}(y^2) \;\leq\; \psi^{k+1}(y^1) .
\end{equation}
Case 3.1: $j=k$. The right hand side of \eqref{stuffnn} becomes
\begin{eqnarray}
&   & \sum_{i=1}^{k}\; \left|\psi^{i}(y^1) - \psi^{i}(y^2)\right| 
 \; + \; \left|x+y^1-\sum_{i=1}^{k}\psi^i(y^1) - \left(x+y^2-\sum_{i=1}^{k}\psi^i(y^2)\right)\right|\\
\nonumber 
& \stackrel{\eqref{l}}{=} & |y^1-y^2| .
\end{eqnarray}
Case 3.2: $j=k+1$. The right hand side of \eqref{stuffnn} turns into
\begin{eqnarray}
\sum_{i=1}^{k}\; \left|\psi^{i}(y^1) - \psi^{i}(y^2)\right| +
\left|\psi^{k+1}(y^1) - \left(x+y^2-\sum_{i=1}^{k}\psi^i(y^2)\right)\right| & & \\
\nonumber \; + \; \left|x+y^1-\sum_{i=1}^{k+1}\psi^i(y^1)\right|  
& \stackrel{\eqref{nice2}}{=} & |y^1-y^2| .
\end{eqnarray}
Case 3.3: $j \geq k+2$. The right hand side of \eqref{stuffnn} becomes
\begin{eqnarray}
\sum_{i=1}^{k}\; \left|\psi^{i}(y^1) - \psi^{i}(y^2)\right| 
\; + \; 
\left|\psi^{k+1}(y^1) - \left(x+y^2-\sum_{i=1}^{k}\psi^i(y^2)\right)\right| 
& & \\
\nonumber + \; 
|\psi^{k+2}(y^1)| + \ldots  + |\psi^{j}(y^1)| 
\; + \; 
\left|x+y^1-\sum_{i=1}^{j}\psi^i(y^1)\right| 
& \stackrel{\eqref{nice2}}{=} & |y^1-y^2| .
\end{eqnarray}
Case 4: Assume $m > j > 0$ and 
\begin{equation}
\sum_{i=1}^{j}\psi^i(y^1) \;\leq\; x+y^1 \;\leq\; \sum_{i=1}^{j+1}\psi^i(y^1) 
\end{equation}
as well as
\begin{equation}
x+y^2 \;\leq\; \psi^1(y^2) .
\end{equation}
The right hand side of \eqref{stuffnn} turns into
\begin{equation}
|\psi^{1}(y^1) - (x+y^2)| +
|\psi^{2}(y^1)| + \ldots + |\psi^{j}(y^1)| +
\left|x+y^1-\sum_{i=1}^{j}\psi^i(y^1)\right| 
\;=\; |y^1-y^2| .
\end{equation}
Case 5: Assume $\sum_{i=1}^{m}\psi^i(y^1) \leq x+y^1$ and  $x+y^2 \leq \psi^1(y^2)$. The right hand side of
\eqref{stuffnn} turns into
\begin{equation}
|\psi^{1}(y^1) - (x+y^2)| +
|\psi^{2}(y^1)| + 
\ldots + 
|\psi^{m}(y^1)| + 
\left|x+y^1-\sum_{i=1}^{m}\psi^i(y^1)\right| 
\;=\; |y^1-y^2| .
\end{equation}
Case 6: Assume $x+y^1 \leq \psi^1(y^1)$ and  $x+y^2 \leq \psi^1(y^2)$. The right hand side of \eqref{stuffnn} 
becomes
\begin{equation}
|(x+y^1) - (x+y^2)| \;=\; |y^1-y^2| .
\end{equation}
The proof for $y^1 \leq y^2$ follows immediately from swapping $y^1$ and $y^2$ with each other
in \eqref{stuffnn}.
\end{proof}


\end{appendix}



\end{document}